\documentclass[a4paper,11pt]{article}

\usepackage{lmodern}  

\usepackage[english]{babel}
\usepackage[T1]{fontenc}
\usepackage[latin1]{inputenc}

\usepackage[margin=1in]{geometry}

\pdfpagewidth=\paperwidth
\pdfpageheight=\paperheight

\newcommand{\menge}[1]{\left\{#1\right\}}
\newcommand{\mengest}[2]{\left\{#1|#2\right\}}
\newcommand{\nor}[1]{\left\|#1\right\|}

\usepackage{amsmath}
\usepackage{amsthm}
\usepackage{amsfonts}
\usepackage[numbers]{natbib}
\usepackage{amssymb}
\usepackage{cancel}
\usepackage{graphicx}
\usepackage{url}
\newtheorem{thm}{Theorem}
\newtheorem{lem}[thm]{Lemma}
\newtheorem{cor}[thm]{Corollary}
\newtheorem{algo}{Algorithm}

\usepackage[lined,algonl,ruled,noend]{algorithm2e}

\title{A Robust AFPTAS for Online Bin Packing with Polynomial Migration}
\author{Klaus Jansen\\
Department of Computer Science\\
Christian-Albrechts-University to Kiel\\\url{ kj@informatik.uni-kiel.de}
 \and 
Kim-Manuel Klein\\
Department of Computer Science\\
Christian-Albrechts-University to Kiel\\
\url{ kmk@informatik.uni-kiel.de}}

\date{}

\begin{document}

\maketitle

\begin{abstract}
In this paper we develop general LP and ILP techniques to find an approximate solution with improved objective
value close to an existing solution.
The task of improving an approximate solution is closely related to a classical theorem of 
Cook et al. \cite{cook1986sensitivity} in the
sensitivity analysis for LPs and ILPs. This result is often applied in designing robust algorithms for online
problems.
We apply our new techniques to the online bin packing problem, where it is allowed to reassign a certain
number of items, measured by the migration factor. The migration factor is defined by
the total size of reassigned items divided by the size of the arriving item.
We obtain a robust asymptotic fully polynomial time approximation scheme (AFPTAS) for the online bin packing
problem with migration 
factor bounded by a polynomial in $\frac{1}{\epsilon}$.
This answers an open question stated by Epstein and Levin \cite{epstein2006robust} in the affirmative.
As a byproduct we prove an approximate variant of the sensitivity
theorem by Cook at el. \cite{cook1986sensitivity} for linear programs.
\end{abstract}

\section{Introduction}
The idea behind robust algorithms is to find solutions of an optimization problem 
that are not only good for a single instance, but also if the instance changes in 
certain ways. Instances change for example due to uncertainty or when new data
arrive. With changing parameters and data, we have the effort to keep as much 
parts of the existing solution as possible, since modifying a solution is often
connected with costs or may even be impossible in practice.
Achieving robustness especially for linear programming (LP) and integer linear programming (ILP)
is thus a big concern and a very interesting research area. Looking at worst case scenarios, how
much do we have to modify a solution if the LP/ILP is changing? There is a result of Cook et al. \cite{cook1986sensitivity}
giving an upper bound for ILPs when changing the right hand side of the ILP.
Many algorithms in the theory of robustness are based on this theorem.

As a concrete application we consider the classical online bin packing problem, 
where items arrive over time and our
objective is to assign these items into as few bins as
possible. The notion of robustness allows to repack a certain number of
already packed items when a new item arrives. 
On the one hand we want to guarantee that we use
only a certain number of additional bins away from the minimum solution
and on the other hand, when a new item arrives, we want to repack as few items as possible.
In the case of offline bin packing 
it is known that unless $\mathcal{P} = \mathcal{NP}$ there is no polynomial time approximation 
algorithm for offline bin packing that produces
a solution better than $\frac{3}{2} \mathit{OPT}$, where $\mathit{OPT}$ is the minimum number
of bins needed. For this reason, the most common way to deal with the inapproximability problem is the 
introduction of the asymptotic approximation ratio. 
The \emph{asymptotic approximation ratio} for an algorithm $A$ is defined to be $\lim_{x \rightarrow \infty} \sup
\{ \frac{A(I)}{\mathit{OPT(I)}} \mid \mathit{OPT(I)} = x \}$.
This leads to the notion of \emph{asymptotic polynomial time approximation schemes (APTAS)}. 
Given an instance of size $n$ and a 
parameter $\epsilon \in (0,1]$, an APTAS has a running time of $\mathit{poly}(n)^{f(\frac{1}{\epsilon})}$ and 
asymptotic approximation
ratio $1+ \epsilon$, where $f$ is an arbitrary function.
An APTAS is called an \emph{asymptotic fully polynomial time approximation scheme (AFPTAS)} if its
running time is polynomial in $n$ and $\frac{1}{\epsilon}$.
The first APTAS for offline bin packing was developed by Fernandez de la Vega \& Lueker \cite{de1981bin},
and Karmakar \& Karp improved this result by giving an AFPTAS \cite{karmarkar1982} (see survey on 
bin packing \cite{coffman1996}).

Since the introduction by Ullman of the classical online bin packing problem \cite{ullman1971}, there has been 
plenty of research 
(see survey \cite {csirik1998}). The best known algorithm has an asymptotic competitive
ratio of $1.58889$ \cite{seiden2002} compared to the optimum in the offline case, 
while the best known lower bound is $1.54037$ \cite{balogh2010}.
Due to the relatively high lower bound of the classical online bin packing problem, there has been effort
to extend the model with the purpose to obtain an improved competitive ratio.
Gambosi et al. \cite{gambosi2000} presented a model where they allow repacking of items. 
They presented an algorithm 
which achieves ratio $1.33$ by moving at most 7 items to a different bin each time
a new item arrives.
The idea of a dynamic packing was developed pretty early by Coffman, Garey and Johnson \cite{coffman1983}.
They developed and analyzed an algorithm for the dynamic bin packing model when arrival and departure of items
are known in advance.
Ivkovic and Lloyd \cite{ivkovic1998} presented an algorithm for dynamic bin packing having ratio $1.25$.
In their model items and their arrival and departure are not known in advance.
The algorithm requires $\mathcal{O}(\log n)$ shifting moves, where a shifting move consists of moving a 
large item or a bundle of small items from one bin to another. 
In another work Ivkovic and Lloyd \cite{ivkovic1997} gave an algorithm which achieves approximation
ratio $1+ \epsilon$ by using amortized $\mathcal{O}(\log n)$ shifting moves.
Concerning reassignment of jobs in scheduling, Albers and Hellwig \cite{albers2012} 
presented an algorithm for online scheduling with minimizing the makespan on $m$ machines. 
The algorithm has a competitive ratio
of $\alpha_m$, where $\alpha_2 = 4/3$ and for $m \rightarrow \infty$ the competitive ratio $\alpha_m$ converges 
to $\approx 1.4659$. The reassignment of jobs is bounded by $\mathcal{O}(m)$. 
They also proved that it is not possible to obtain an algorithm with better competitive ratio than $\alpha_m$
without reassigning $\Theta (n)$ jobs.

The model we follow is the notion of robustness. Introduced by Sanders et al. \cite{sanders2009} it allows 
repacking of arbitrary items while the number of items that are being repacked is limited. To give a measure 
on how many items are allowed to
be repacked Sanders et al. \cite{sanders2009} defined the \emph{migration factor}.
It is defined by the complete size of all moved items divided by the size of the arriving one.
An (A)PTAS is called \emph{robust} if its migration factor is of the size $f(\frac{1}{\epsilon})$, where $f$ is
an arbitrary function that only depends on $\frac{1}{\epsilon}$.
Since the promising introduction of robustness, several robust algorithms have been developed.
Sanders et al. \cite{sanders2009} found a robust PTAS for the online scheduling problem  
on identical machines, where the goal is to minimize the makespan. The robust PTAS has
constant but exponential migration factor $2^{\mathcal{O}(\frac{1}{\epsilon} \log^2 \frac{1}{\epsilon})}$.
In case of bin packing
Epstein and Levin \cite{epstein2006robust} developed a robust APTAS for the classical bin packing problem
with migration factor $2^{\mathcal{O}(\frac{1}{\epsilon^2} \log \frac{1}{\epsilon})}$ and running time
double exponential in $\frac{1}{\epsilon}$.
In addition they proved that there is no optimal online algorithm with a constant migration factor.
Furthermore, Epstein and Levin \cite{epsteinu} showed that the robust APTAS for bin packing can 
be generalized to packing d-dimensional cubes into a minimum number of unit cubes.
Recently Epstein and Levin \cite{epstein2011} also designed a robust 
algorithm for preemptive online
scheduling of jobs on identical machines, where the corresponding offline problem is polynomial solvable.
They presented an algorithm with migration factor
$1-\frac{1}{m}$ that computes an optimal solution whenever a new item arrives.
Skutella and Verschae \cite{skutella2010} studied the problem of maximizing the minimum load given 
$n$ jobs and $m$ machines. They proved that there is no robust PTAS for this machine covering problem.
On the positive side, they gave a robust PTAS for the machine covering problem in the case that
migrations can be reserved for a later timestep. The algorithm has an amortized migration factor of
$2^{\mathcal{O}(\frac{1}{\epsilon} \log^2 \frac{1}{\epsilon})}$.
\subsection{Our Results:}
An online algorithm is called \emph{fully robust} if its migration factor is bounded by
$p(\frac{1}{\epsilon})$, where $p$ is a polynomial in $\frac{1}{\epsilon}$.
The purpose of this paper is to give methods to develop fully robust
algorithms. In Section 2 we develop a theorem for a given linear program (LP)
$\min \mengest{\nor{x}_1}{Ax \geq b, x \geq 0 }$.
Given an approximate solution $x'$ 
with value $(1+ \delta)\mathit{LIN}$ (where $\mathit{LIN}$ is the minimum objective value of the LP)
and a parameter $\alpha \in (0,\delta \mathit{LIN} ]$, we prove the existence of an improved solution $x''$
with value $(1+ \delta)\mathit{LIN} - \alpha$ and distance $\nor{x'' - x'}_1 \leq \alpha(2/ \delta + 2)$.
In addition, for a given fractional solution $x'$ and corresponding integral solution $y'$, the existence
of an improved integral solution $y''$ with $\nor{y''-y'}_1 = \mathcal{O}(\frac{\alpha+m}{\delta})$ is shown
(where $m$ is the number of rows of $A$). Since both results are constructive, we propose also algorithms
to compute such improved solutions.
Previous robust online algorithms require an optimum solution of the corresponding ILP and use a sensitivity theorem
by Cook et al. \cite{cook1986sensitivity}.
This results in an exponential migration factor in $\frac{1}{\epsilon}$
(\cite{epstein2006robust, epsteinu, skutella2010, sanders2009}).
In contrast to this we consider approximate solutions of the corresponding LP relaxations and are able
to use the techniques above to improve the fractional and integral solutions.
Furthermore we also prove an approximate version of a sensitivity theorem for LPs with modified right
hand side $b$ and $b'$.
During the online algorithm the number of non-zero variables increases from step to step and would result
in a large additive term. To avoid this we present algorithms in Section 3 to control the number
of non-zero variables of the LP and ILP solutions. We can bound the number of non-zero variables and the
additive term by $\mathcal{O}(\epsilon \mathit{LIN}) + \mathcal{O}(\frac{1}{\epsilon^2})$.
In Section 4 we present the fully robust AFPTAS for the robust bin packing problem. We use a modified
version of the clever rounding techniques of Epstein and Levin \cite{epstein2006robust}.
This rounding technique is used to round the incoming items dynamically and control the number of item sizes.
One difficulty is that we use approximate solutions of the LP.
During the online algorithm items are rounded to different values and are shifted across different 
rounding groups. We show how to embed the rounded instance into another rounded instance that fulfills
several invariants.
By combining the dynanic rounding and the algorithm to get improved solutions of the LP and ILP, we are able
to obtain a fully robust AFPTAS for the online bin packing problem.
The algorithm has a migration factor of 
$\mathcal{O}(1/\epsilon^4)$ (or $\mathcal{O}(1/\epsilon^3)$ if the size of the arriving item is $\Omega(1)$) and 
running time polynomial in $\frac{1}{\epsilon}$ and $t$, where $t$ is the number of arrived items. 
This resolves an open question of Epstein and Levin \cite{epstein2006robust}. We believe that our techniques
can be used for other online problems like 2D strip packing, scheduling moldable tasks, resource constrained
scheduling and multi-commodity flow problems to obtain online algorithms with low migration factors.
\section{Robustness of approximate LPs}
We consider a matrix $A \in \mathbb{R}_{\geq 0}^{m\times n}$, a vector $b \in \mathbb{R}_{\geq0}^{m}$ and a cost vector
$c\in \mathbb{R}_{\geq0}^{n}$.
The goal in a \emph{linear program} (LP) is to find a $x\geq0$ with $Ax \geq b$ such that the
\emph{objective value} $c^T x$ is minimal.
We say $x^\mathit{OPT}$ is an optimal solution if $c^T x^\mathit{OPT} = \min \menge{c^T x | Ax \geq b, x \geq 0 }$ 
and we define $\mathit{LIN} = c^T x^\mathit{OPT}$. In general we suppose that the objective function of a solution
is positive and hence $\mathit{LIN} >0$.
We say $x'$ is an approximate solution with approximation ratio $(1+\delta)$ for some 
$\delta \in (0,1]$ if $\nor{x'}_1 \leq (1+\delta)\mathit{LIN}$. For the most part of the paper we will assume that
$c^T = (1,1, \ldots ,1)$ and therefore $c^T x^{\mathit{OPT}} = \nor{x^{\mathit{OPT}}}_1 = \mathit{LIN}$.
The following theorem is central. Given an approximate solution $x'$, we want to improve its
approximation by some constant. But to achieve robustness we have to maintain most
parts of $x'$. We show that by changing $x'$ by size of $\mathcal{O}(\frac{\alpha}{\delta})$, 
we can improve the approximation by a constant $\alpha$.
\begin{thm}
  \label{thm1}
  Consider the LP $\min \menge{c^T x | Ax \geq b, x
  \geq 0 }$ and an approximate solution $x'$ with $c^T x'
  = (1 + \delta) \mathit{LIN}$ for some $\delta >0$.  For every positive $\alpha \leq \delta \mathit{LIN}$ there 
  exists a solution $x''$ with objective value of at most $c^T x'' \leq (1 +
  \delta) \mathit{LIN}- \alpha$ and distance $\nor{x' - x''}_1 \leq \alpha (1/ \delta +
  1) \frac{\nor{x'}_1+ \nor{x^{\mathit{OPT}}}_1}{c^T x'}$.
  If $c^T = (1,1, \ldots ,1)$ then $\nor{x' - x''}_1 \leq 2 \alpha (1/ \delta + 1)$.
\end{thm}
\begin{proof}
   We prove feasibility of the following \ref{form:lp1}. 
  \begin{align*}
    \label{form:lp1}
    \tag{LP 1}
      Ax &\geq b\\
      x &\geq 0\\
      x &\geq x' -\frac{\alpha(1/\delta +1)}{c^T x'}x'\\
      x &\leq x' +\frac{\alpha(1/\delta +1)}{c^T x'}x^{\mathit{OPT}}\\
      c^T x & \leq (1+ \delta) \mathit{LIN} - \alpha
  \end{align*}
  The assumption $\alpha \leq \delta \mathit{LIN}$ implies $c^T x'= (1+\delta) \mathit{LIN} \geq \alpha(1/ \delta +1)$
   and hence $\frac{\alpha(1/\delta +1)}{c^T x'} \leq 1$.
  Suppose that \ref{form:lp1} is feasible and has a solution $x''$. Due to constraints 3 and 4 the distance between 
 $x''$ and $x'$ can be bounded. Components $x''_i$ with $x'_i > 0$ may be smaller compared to $x'$ by
  $\alpha(1/ \delta +1) \frac{x'_i}{c^T x'}$, while components $x''_i$ with
  $x_i^{\mathit{OPT}} >0$ may be larger than $x'_i$ by $\alpha(1/ \delta + 1)
  \frac{x^{\mathit{OPT}}_i} {{c^T x'}}$.
  In the worst case, $x'$ and $x^{\mathit{OPT}}$ have no common non-zero entries 
  and hence, $\nor{x' - x''}_1 \leq \alpha(1/ \delta +1) (\sum_i \frac{x'_i} {{c^T x'}} + 
  \sum_i \frac{x^{\mathit{OPT}}_i} {{c^T x'}}) = \alpha (1/ \delta +
  1) \frac{\nor{x'}_1+ \nor{x^{\mathit{OPT}}}_1}{c^T x'}$.
	If $c^T = (1,1, \ldots ,1)$ then $\frac{\nor{x'}_1}{c^T x'} = 1$ and $\frac{\nor{x^{\mathit{OPT}}}_1}{c^T x'}=
	\frac{1}{(1+\delta)} < 1$. Therefore  $\alpha (1/ \delta + 1) \frac{\nor{x'}_1+ \nor{x^{\mathit{OPT}}}_1}{c^T x'}
	< 2 \alpha (1/ \delta + 1)$.
  It remains to prove feasibility of \ref{form:lp1}.
  We construct a solution $x''$ by 
  $x'' = (1- \frac{\alpha(1/\delta +1)}{c^T x'})x' + \frac{\alpha(1/\delta +1)}{c^T x'}x^{\mathit{OPT}}$.
	We prove that each constraint of \ref{form:lp1}	is satisfied for $x''$. Note that $x'' \geq 0$ since
	$\frac{\alpha(1/\delta +1)}{c^T x'} \leq 1$.
	Constraint 3 is fulfilled since $x'' =  x' -\frac{\alpha(1/\delta +1)}{c^T x'}x' + 
	\frac{\alpha(1/\delta +1)}{c^T x'}x^{\mathit{OPT}} \geq x' -\frac{\alpha(1/\delta +1)}{c^T x'}x'$ and
	constraint 4 is fulfilled since $x'' = x' -\frac{\alpha(1/\delta +1)}{c^T x'}x' + 
	\frac{\alpha(1/\delta +1)}{c^T x'}x^{\mathit{OPT}} \leq x' + \frac{\alpha(1/\delta +1)}{c^T x'}x^{\mathit{OPT}}$.
  Feasibility for $x''$ follows from 
  $A x'' = Ax'(1 - \frac{\alpha(1/\delta +1)}{c^T x'}) + Ax^{\mathit{OPT}}\frac{\alpha(1/\delta +1)}{c^T x'}
  \stackrel{c^T x' \geq \alpha(1/\delta +1)}{\geq} b(1- \frac{\alpha(1/\delta +1)}{c^T x'}) + b \frac{\alpha(1/\delta +1)}{c^T x'} =b$.
  The objective value of $x''$ is bounded by
  $c^T x'' = c^T x' - \frac{\alpha(1/\delta +1)}{c^T x'}(c^T x' - c^T x^{\mathit{OPT}}) = 
  c^T x' - \alpha(1/\delta +1) + \frac{\alpha(1/\delta +1)}{(1+\delta) \mathit{LIN}} \mathit{LIN}
  = c^T x' - \alpha(1/\delta +1) + \alpha (1/\delta)
  = c^T x' - \alpha$.
\end{proof}
From here on and for the rest of the paper we suppose that $c^T = (1,1, \ldots ,1)$.\\
{\bf Remark 1:\\}
Suppose $x'$ has approximation ratio $\nor{x'}_1 = (1+\delta') \mathit{LIN}$ for some $\delta' > 0$.
By Theorem \ref{thm1} the following LP is feasible if $\alpha' \leq \delta' \mathit{LIN}$.
\begin{align*}
      Ax &\geq b\\
      x &\geq 0\\
      x &\geq x' - \alpha'(1/\delta' +1) \frac{x'}{\nor{x'}_1}\\
      x &\leq x' + \alpha'(1/\delta' +1) \frac{x^{\mathit{OPT}}}{\nor{x'}_1}\\
      \sum x_i &\leq (1+ \delta') \mathit{LIN} - \alpha'
\end{align*}
Setting $\alpha' = \alpha \frac{(1/\delta+1)}{(1/\delta'+1)}$ for some $\delta >0$ yields feasibility for the following
LP assuming $\nor{x'}_1 = (1+ \delta')\mathit{LIN}$.
\begin{align*}
		\label{form:lpr}
    \tag{LP *}
      Ax &\geq b\\
      x &\geq 0\\
      x &\geq x' - \alpha(1/\delta+1) \frac{x'}{\nor{x'}_1}\\
      x &\leq x' + \alpha(1/\delta +1) \frac{x^{\mathit{OPT}}}{\nor{x'}_1}\\
      \sum x_i &\leq (1+ \delta') \mathit{LIN} - \alpha'
\end{align*}
Here, we use $\alpha' (1/\delta'+1)= \alpha \frac{1/\delta+1}{1/\delta'+1} (1/\delta'+1)
= \alpha(1/\delta +1)$. The condition that $\alpha' \leq \delta' \mathit{LIN}$ is equivalent to the condition that
$\nor{x'}_1 \geq \alpha (1/\delta +1)$ since $\nor{x'}_1 = (1+ \delta')\mathit{LIN}$ and $\alpha (1/\delta +1) =
\alpha' (1/ \delta' +1)$. \qed

{\bf Remark 2:\\}
In many cases, we do not know the exact approximation ratio $\nor{x'}_1 = (1+ \delta') \mathit{LIN}$ but the 
approximation
guarantee $\nor{x'}_1 \leq (1+ \delta) \mathit{LIN}$ for some $\delta \geq \delta'$.
Assuming $\nor{x'}_1 \geq \alpha (1/\delta +1)$ we can use feasibility of \ref{form:lpr} to prove the existence of
a solution $x''$ with $\nor{x''}_1 \leq (1+ \delta) \mathit{LIN} - \alpha$ and $\nor{x''-x'}_1 \leq 2\alpha(1/\delta +1)$.
The distance $\nor{x''-x'}_1 \leq 2\alpha(1/\delta +1)$ follows again easily from constraints 3 and 4 of \ref{form:lpr}.
We derive the aimed objective value $\nor{x''}_1$ from the last constraint of \ref{form:lpr}:
$\nor{x''}_1 \leq (1+ \delta')\mathit{LIN} - 
\alpha \frac{1/\delta+1}{1/\delta'+1} = (1+ \delta)\mathit{LIN} - (\delta - \delta')\mathit{LIN} - 
\alpha(1/\delta+1) \frac{1}{1/\delta'+1} \stackrel{(1+ \delta')\mathit{LIN} \geq  \alpha(1/ \delta +1) }{\leq} 
	(1+ \delta) \mathit{LIN} -
       \alpha(1/\delta +1)\frac{\delta - \delta'}{1+ \delta'} - \alpha(1/\delta+1) \frac{\delta'}{1+\delta'} 
       =(1+ \delta) \mathit{LIN} -  \alpha(1/\delta +1)\frac{\delta}{1+ \delta'}
       = (1+ \delta) \mathit{LIN} -  \alpha \frac{1+ \delta}{1+ \delta'} \stackrel{\delta'\leq \delta}{\leq}
       (1+ \delta) \mathit{LIN} -  \alpha$.
This proves, that it suffices to know an upper bound for the approximation to obtain an improved solution $x''$.\qed

Of course, one major application of Theorem \ref{thm1} is to improve the approximation. 
But we can also apply Theorem \ref{thm1} to obtain a variant of the theorem of
Cook et al. \cite{cook1986sensitivity} for the sensitivity analysis of
an LP. Consider the following problem: Let $x'$ be a solution of $\min
\mengest{ \nor{x}_1}{Ax \geq b', x \geq 0 }$.  Find a solution $x''$ for $\mathit{LIN}_2 = 
\min \mengest{ \nor{x}_1}{Ax \geq b'', x \geq 0 }$ with changed right hand side such that
$\nor{x''-x'}_1$ is small. A theorem of Cook
et al. \cite{cook1986sensitivity} states that there exists a $x''$
satisfying the LP and $\nor{x''-x'}_\infty \leq n \Delta
\nor{b''-b'}_\infty$, where $\Delta$ is the largest subdeterminant of
$A$. This result is not satisfying if $\Delta$ and $n$ are too big,
especially if they are exponential in $m$. By letting loose of optimal solutions we obtain
a corollary that is much more appropriate to derive fully robust algorithms.
In contrary to the theorem of Cook et al. \cite{cook1986sensitivity} the amount of change in the solution does not depend on the determinant 
nor on the dimensions
of $A$ but on the approximation ratio of the solution.

\begin{cor}
 Consider the linear program $LP^1$ defined by $\min \menge{\nor{x}_1 | Ax \geq b', x
  \geq 0 }$ and an approximate solution $x'$ with $\nor{x'}_1
  \leq (1 + \delta) \mathit{LIN}_1$ ($\delta > 0$) and $\nor{x'}_1 \geq \alpha (1/ \delta +1)$. There exists a solution
  $x''$ of $LP^2$ defined by $\min \menge{\nor{x}_1 | Ax \geq b'', x  \geq 0 }$ with $\nor{x''}_1
  \leq (1 + \delta) \mathit{LIN}_2$ such that the distance $\nor{x''-x'}_1 \leq (\frac{2}{\delta} +7)\nor{\frac{b''-b'}{c}}_1$
  where $c_i =\max_j A_{ij}$ and $\frac{b''-b'}{c}$ is a vector having components $\frac{b''_i-b'_i}{c_i}$.
\end{cor}
\begin{proof}
	Suppose there is only one index $i$ where $b'_i \neq b''_i$. Consider the 2 cases:\\
	{\bf Case 1:} $b'_i < b''_i$ We increase $x'_j$ by $\frac{b''_i -b'_i}{c_i}$, where $j$ is the index with the maximum entry in row $i$.
	This way we make sure that the so modified $x'$ covers the larger $b''_i$ since now $(Ax')_i = b''_i$.
	Since we simply increase $x'$ to cover the larger $b''$ we may worsen the approximation by an additive term of
	at most $\frac{b''_i -b'_i}{c_i}$.\\
	{\bf Case 2:} $b'_i \geq b''_i$. In this case we do not modify component $i$ of $x'$, but since a smaller $b''_i$ has to
	be covered the optimal value of a solution may decrease. 
	Let $\mathit{LIN}_1$ be the optimal value of $LP^1$ and let $\mathit{LIN}_2$ be the optimal value of $LP^2$.
	The inequality $\mathit{LIN}_2 < \mathit{LIN}_1 - \frac{b''_i-b'_i}{c_i}$ leads to a contradiction since
	we can increase an optimal solution $x^{\mathit{LIN}_1}$ of $LP^1$ to cover the lager $b'$ like 
	we did in case 1. Modifying 
	$x^{\mathit{LIN}_1}$ this way would lead to a smaller optimal solution.
	Therefore the optimal solution of $LP^1$ can not decline by more than $\frac{b''_i -b'_i}{c_i}$.
	Using $x'$ as an approximate solution for $LP^2$ yields therefore $\nor{x'}_1 \leq (1 + \delta) (\mathit{LIN}_2 
	+ \frac{b''_i -b'_i}{c_i})
	= (1 + \delta) \mathit{LIN}_2 + (1 + \delta)\frac{b''_i -b'_i}{c_i}$.
	
	Iterating over all components $1 \leq i \leq m$ and changing the solution according to the cases would result 
	in an approximate solution of at most
	$(1 + \delta) \mathit{LIN} + (1 + \delta) \nor{\frac{b''-b'}{c}}_1$.
	Using Remark 2 with $\alpha = (1 + \delta)\nor{\frac{b''-b'}{c}}_1$
	guarantees the existence of a solution $\hat{x}$ for $LP^2$ having value $(1 + \delta)\mathit{LIN} - \alpha$ and 
	$\nor{\hat{x}-x'}_1 \leq (2/ \delta +  2) \alpha = (2/ \delta +  2) (1 + \delta)\nor{\frac{b''-b'}{c}}_1
	\leq (2/ \delta +  6)\nor{\frac{b''-b'}{c}}_1$.
	Modifying $\hat{x}$ according to the cases yields therefore a solution $x''$ with $\nor{x''}_1 \leq 
	(1 + \delta) \mathit{LIN}_2$. For the distance between $x''$ and $x'$ we get $\nor{x''-x'}_1 \leq \nor{\hat{x}-x'}_1 + 
	\nor{\frac{b''-b'}{c}}_1 \leq (\frac{2}{\delta} +7)\nor{\frac{b''-b'}{c}}_1$.
\end{proof}
Note that if $A$ is an integral matrix without zero rows, each component $c_i$ is at least $1$.

\section{Algorithmic Use}
Let $x'$ be an approximate solution of the LP with $\nor{x'}_1 \leq (1+ \delta) \mathit{LIN}$. 
In Theorem \ref{thm1}, we have proven the existence of a solution $x''$ near $x'$ with 
$\nor{x''}_1 \leq (1+ \delta) \mathit{LIN} - \alpha$.
We are looking now for algorithmic ways to calculate this improved solution $x''$.
We present two algorithms that basically rely on solving an LP.
According to \ref{form:lpr} 
we split $x'$ into a fixed part $x^{fix}$ and a variable part $x^{var}$.
The variable part is defined according to \ref{form:lpr} by $x^{var} =
\frac{\alpha(1 / \delta +1)}{\nor{x'}}x'$ and the fixed part by $x^{fix} = x' -
x^{var}$. By assigning the variable part in a better way, Theorem
\ref{thm1} and Remark 2 state under the assumption that $\nor{x'}_1 \geq \alpha(1/ \delta +1)$ that 
we can improve the objective value by $\alpha$. 
We denote with $b^{var} = b - A(x^{fix})$ the part which has to be reassigned. The
algorithm works as follows:
\begin{algo} \label{alg1}
\ 
  \begin{enumerate}
\item Set $x^{var} := \frac{\alpha(1 / \delta +1)}{\nor{x'}}x'$, $x^{fix} := x' - x^{var}$
  and $b^{var} := b - A(x^{fix})$
\item Solve the LP $\hat{x} = \min \mengest{\nor{x}_1}{Ax \geq b^{var}, x
    \geq 0 }$
\item Generate a new solution $x'' = x^{fix} + \hat{x}$
  \end{enumerate}
\end{algo}
If $\hat{x}$ is a basic feasible solution, compared to $x'$, our new solution $x''$ has up to $m$ additional
non-zero components.
\begin{thm}
	Given solution $x'$ with $\nor{x'}_1 \leq (1+\delta)\mathit{LIN}$ and $\nor{x'}_1 \geq \alpha(1/ \delta +1)$.
	Algorithm \ref{alg1} returns a feasible solution $x''$ with $\nor{x''}_1 \leq (1+\delta)\mathit{LIN} - \alpha$ and the 
	distance between $x'$ and $x''$ is $\nor{x''-x'}_1 \leq 2 \alpha(1/ \delta +1)$.
\end{thm}
\begin{proof}
Solution $x''$ is feasible because $A(x'') = A(x^{fix} + \hat{x}) = A(x^{fix}) + A(\hat{x}) \geq A(x^{fix}) + b^{var} = b$.
For the approximation we use Remark 2, which guarantees the existence of a solution 
with objective value 
$\leq (1+\delta)\mathit{LIN} - \alpha$ by leaving the part $x^{fix}$ of the solution $x'$ unchanged.
The unchanged part $x^{fix}$ is defined by using the lower bounds of \ref{form:lpr}, 
$x'' \geq x'- \alpha (1/ \delta +1) \frac{x'}{\nor{x'}_1}$. Placing
$\hat{x}$ optimally leads therefore to the aimed approximation. Since the $\nor{x^{var}}_1$ and $\nor{\hat{x}}_1$ are bounded by
$\alpha(1/ \delta +1)$ the worst possible distance between $x'$ and $x''$ is $\nor{x''-x'}_1 = 
\nor{(x^{fix} + \hat{x})- (x^{fix} + x^{var})}_1 = \nor{\hat{x}-x^{var}}_1 \leq \nor{\hat{x}}_1 + \nor{x^{var}}_1
\leq 2 \alpha(1/ \delta +1)$.
\end{proof}
In Algorithm \ref{alg1} we use an optimal LP solver as a subroutine. In many cases, like for example bin packing,
the corresponding LP relaxation is hard to solve and
the running time
for computing an optimal solution is very high. For the following algorithm it is sufficient to compute the LP
approximately, which in general can be performed more efficiently. We assume that $\nor{x'}_1 \geq 2 \alpha 
(1/ \delta +1)$
because the double amount has to be reassigned to achieve the same improvement in the approximation as in Algorithm \ref{alg1}.
\begin{algo}\label{alg2}
\ 
\begin{enumerate}
\item Set $x^{var} := \frac{2 \alpha(1 / \delta +1)}{\nor{x'}}x'$, $x^{fix} := x' - x^{var}$
  and $b^{var} := b - A(x^{fix})$
\item Solve $\hat{x} = \min \mengest{\nor{x}_1}{Ax \geq b^{var}, x \geq 0 }$
  approximately with ratio $(1+\delta / 2)$
\item If $\nor{x^{fix} + \hat{x}}_1 < \nor{x'}_1$ set $x'' := x^{fix} + \hat{x}$
else $x'' = x'$.
\end{enumerate}
\end{algo}
\begin{thm} \label{thm2}
	Let $x'$ be a solution with $\nor{x'}_1 \leq (1+\delta) \mathit{LIN}$ and $\nor{x'}_1 \geq 2 \alpha(1/ \delta +1)$.
  Then Algorithm \ref{alg2} returns a feasible solution $x''$ with approximation guarantee $(1+ \delta) \mathit{LIN} - \alpha$ 
  and $\nor{x''-x'}_1 \leq 4 \alpha (1/ \delta +1)$.
\end{thm}
\begin{proof}
	The property that $\nor{x''-x'}_1 \leq 4 \alpha (1/ \delta +1)$ follows by Theorem 3 and the fact,
	that $x^{var}$ has the double size $2 \alpha(1/ \delta +1)$ compared to $x^{var}$ defined in Algorithm \ref{alg1}. 
	Furthermore we have to show that at the end of the algorithm $\nor{x''}_1 \leq (1+ \delta) \mathit{LIN} - \alpha$.
	Suppose $\nor{x'} = (1+ \delta') \mathit{LIN}$ for some $\delta' \leq \delta$. Using the assumption 
	$2 \alpha(1/ \delta +1) \leq \nor{x'}_1 \leq (1+ \delta)\mathit{LIN}$ implies that
	$2 \alpha \leq \frac{(1+ \delta)\mathit{LIN}}{(1/\delta +1)} = \delta \mathit{LIN}$.
	Consider the case that $2 \delta' \leq \delta$. In this case $x''$ has the aimed approximation since 
	$\nor{x''}_1 \leq \nor{x'}_1 = (1+ \delta') 
	\mathit{LIN} \leq (1+ \delta) \mathit{LIN} - \delta/2 \mathit{LIN} \leq (1+ \delta) \mathit{LIN} - \alpha$ using
	$2\alpha \leq \delta \mathit{LIN}$.
	Thus in the following we assume $\delta \leq 2 \delta'$.
	Suppose we solve the LP in step 2 optimally. In this case, Algorithm \ref{alg2} is identical to Algorithm \ref{alg1} using improvement
	of $2 \alpha$. By feasibility of \ref{form:lpr} we know there exists a solution $\bar{x}''$
	with 
	$\nor{\bar{x}''} \leq (1+ \delta') \mathit{LIN} - 2\alpha(\frac{1/ \delta +1}{1/ \delta'+1})$.
	This implies, that an optimal solution $\hat{x}^{\mathit{OPT}}$ of the LP 
	$\min \mengest{\nor{x}_1}{Ax \geq b^{var}, x \geq 0 }$ is of size $\nor{\hat{x}^{\mathit{OPT}}}_1
	\leq \nor{x^{var}}_1 - 2\alpha \frac{1/\delta+1}{1/\delta'+1} = 2 \alpha(1/ \delta +1) -
	2\alpha \frac{1/\delta+1}{1/\delta'+1}$. Solving the LP approximately with ratio $(1+\delta / 2)$, 
	solution $\hat{x}$ has an additional term $\delta / 2 \nor{\hat{x}^{\mathit{OPT}}}_1$. The value of 
	$\nor{\hat{x}}_1$ is therefore bounded by
	$\nor{\hat{x}}_1 \leq \nor{\hat{x}^{\mathit{OPT}}}_1 + \delta / 2 \nor{\hat{x}^{\mathit{OPT}}}_1 
	\leq \nor{x^{var}}_1 - 2\alpha \frac{1/\delta+1}{1/\delta'+1} + 
	\alpha(1+ \delta) - \alpha (\frac{1+ \delta}{1/ \delta'+1})$.
	Finally this results in the approximation for $x^{fix} + \hat{x}$ as follows.
\begin{align*}
		\nor{x^{fix} + \hat{x}}_1 & = \nor{x'}_1 - \nor{x^{var}}_1 + \nor{\hat{x}}_1\\
       = &(1+ \delta') \mathit{LIN} -2 \alpha (\frac{1/ \delta +1}{1/ \delta'+1}) + \alpha(1+ \delta) 
       - \alpha (\frac{1+ \delta}{1/ \delta'+1}) \\
       = &(1+ \delta) \mathit{LIN} - (\delta - \delta') \mathit{LIN} -2 \alpha (\frac{1/ \delta +1}{1/ \delta'+1})
        + \alpha(1+ \delta) 
       - \alpha (\frac{1+ \delta}{1/ \delta'+1}) \\
       \stackrel{\mathit{LIN} \geq 2 \alpha / \delta}{\leq} &(1+ \delta) \mathit{LIN} -
       2 \alpha(\frac{\delta - \delta'}{\delta}) + \alpha(1+ \delta)  -2 \alpha (\frac{1/ \delta +1}{1/ \delta'+1})  
       - \alpha (\frac{1+ \delta}{1/ \delta'+1})\\
       = & (1+ \delta) \mathit{LIN} + 
       \alpha(\frac{-2 \delta + 2 \delta' + \delta + \delta^2}{\delta}) -2 \alpha (\frac{1/ \delta +1}{1/ \delta'+1})  
       - \alpha (\frac{1+ \delta}{1/ \delta'+1})\\
       = & (1+ \delta) \mathit{LIN} - \alpha + \alpha(\frac{2 \delta' + \delta^2}{\delta}) 
       -2 \alpha (\frac{1/ \delta +1}{1/ \delta'+1}) - \alpha (\frac{1+ \delta}{1/ \delta'+1})\\
       = & (1+ \delta) \mathit{LIN} - \alpha 
        +\alpha(\frac{2+ 2\delta' +\delta^2/\delta' + \delta^2 -2 -2\delta - \delta - \delta^2}{\delta(1/ \delta'+1)})\\
       = & (1+ \delta) \mathit{LIN} - \alpha + 
       \alpha(\frac{2\delta' +\delta^2/\delta' -3 \delta}{\delta(1/ \delta'+1)})\\
       = & (1+ \delta) \mathit{LIN} - \alpha + \alpha(\frac{(\delta - \delta') 
       (-2+ \frac{\delta}{\delta'})}{\delta(1/ \delta'+1)})\\
       \leq & (1+ \delta) \mathit{LIN} - \alpha
\end{align*}
The last inequality holds because $\alpha(\frac{(\delta - \delta') (-2+ \frac{\delta}{\delta'})}{\delta(1/ \delta'+1)}) 
\leq 0$ since $\delta - \delta' \geq 0$ and $-2+ \frac{\delta}{\delta'} \leq 0 \Leftrightarrow  \delta \leq 2\delta'$.
	By the last step of the algorithm we know that $\nor{x''}_1 \leq  \nor{x^{fix} + \hat{x}}_1$ and thus
	$\nor{x''}_1 \leq (1+ \delta)\mathit{LIN} - \alpha$.
\end{proof}
In some cases we may not want to get a guaranteed approximation, but a guarantee that our solution $x'$ is getting
smaller by some $\alpha$. This works if the approximation ratio of $x'$ is worse than $(1+\delta)$. 
The following corollary states, that if we use Algorithm \ref{alg2} on a solution $x'$ with 
$\nor{x'}_1 = (1+ \delta')\mathit{LIN}$ for some $\delta' \geq \delta$ the objective function of our new solution 
$x''$ decreases by at least $\alpha$.
\begin{cor}\label{cor5}
	Let $\nor{x'}_1 = (1+ \delta')\mathit{LIN}$ for some $\delta' \geq \delta$ and 
	$\nor{x'}_1 \geq 2 \alpha(1/ \delta +1)$. Then Algorithm \ref{alg2} returns a solution
	$x''$ with $\nor{x''}_1 \leq \nor{x'}_1 - \alpha = (1+ \delta')\mathit{LIN} - \alpha$ and 
	$\nor{x''-x'}_1 \leq 4 \alpha (1/ \delta +1)$.
\end{cor}
\begin{proof}
	Suppose like in the proof of Theorem 4 that we solve the LP in step 2 optimally. In this case, Algorithm \ref{alg2} is 
	identical to Algorithm \ref{alg1} using improvement
	of $2 \alpha$ and therefore by feasibility of \ref{form:lpr} we know that it returns a solution $\bar{x}''$
	with 
	$\nor{\bar{x}''} \leq (1+ \delta') \mathit{LIN} - 2\alpha(\frac{1/ \delta +1}{1/ \delta'+1})$.
	An optimal solution $\hat{x}^{\mathit{OPT}}$ of the LP
	$\min \mengest{\nor{x}_1}{Ax \geq b^{var}, x \geq 0 }$ is therefore of size $\nor{\hat{x}^{\mathit{OPT}}}_1
	\leq \nor{x^{var}}_1 - 2\alpha \frac{1/\delta+1}{1/\delta'+1} = 2 \alpha(1/ \delta +1) -
	2\alpha \frac{1/\delta+1}{1/\delta'+1}$. Since we actually solve the LP approximately with ratio $(1+\delta / 2)$, 
	solution $\hat{x}$ has an additional term of $\delta / 2 \nor{\hat{x}^{\mathit{OPT}}}_1$ and the value is therefore 
	bounded by
	$\nor{\hat{x}}_1 \leq \nor{x^{var}}_1 - 2\alpha \frac{1/\delta+1}{1/\delta'+1} + 
	\alpha(1+ \delta) - \alpha (\frac{1+ \delta}{1/ \delta'+1})$ according to the proof of Theorem 5.
	By construction of $x''$ we get $\nor{x''}_1 \leq \nor{x^{fix} + \hat{x}}_1 = 
(1+ \delta') \mathit{LIN} -2 \alpha (\frac{1/ \delta +1}{1/ \delta'+1}) + \alpha(1+ \delta) 
       - \alpha (\frac{1+ \delta}{1/ \delta'+1})$. Since $\delta' \geq \delta$ we know that
       $-2 \alpha (\frac{1/ \delta +1}{1/ \delta'+1}) \leq -2 \alpha$ and that $\alpha(1+ \delta) 
       - \alpha (\frac{1+ \delta}{1/ \delta'+1}) \leq \alpha$. Hence $\nor{x''}_1 \leq  
       (1+ \delta')\mathit{LIN} - \alpha$.
\end{proof}
\subsection{Integer Programming}
\label{sec:integer}
In this section we discuss how we can apply results from the previous sections to integer programming.
Consider a fractional solution $x'$ of the LP and a corresponding integral solution $y'$. 
By rounding each component $x'_i$ up to the next integer value, it is easy to get a feasible
integer solution $y'$ with an additional additive term $\nor{y'}_1 \leq \nor{x'}_1 +C$, where $C$ is the number of 
non-zero components.
We can apply any of the previous algorithms to $x'$ to get an improved solution $x''$.
But our actual goal is to find a corresponding integer solution $y''$
with improved objective value $\nor{y''}_1 \leq (1+\delta) \mathit{LIN} +C - \alpha$ such that the 
distance between $y''$ and $y'$ is small.
In the following we present two algorithms that compute a suitable $y''$ with improved objective value and small
distance between $y''$ and $y'$.
Note that the straight forward approach to simply round up each component $x''_i$ leads to a distance between
$y''$ and $y'$ that depends on $C$ and hence (depending on the LP) is too high.
Designing the algorithms, there seems to be some trade off between the number of non-zero components and the
distance between the integer solutions $y'$ and $y''$. The first algorithms tries to minimize the distance 
between $y'$ and $y''$ while
the second guarantees better approximation of $\nor{y'}_1$ and $\nor{y''}_1$ while the distance between them increases.
The existence of an algorithm combining both good properties, low distance and good approximation guarantee of 
$y'$ and $y''$, is an interesting question.

In Algorithm \ref{alg3} we focus on how much components of $x'$ need to be reduced to achieve the improved approximation guarantee.
This defines the migration factor in robust bin packing. The actual worst case distance between $y''$ and $y'$ is 
larger and however can only
be bounded by $\mathcal{O}(m + 1/ \delta)$. Like in the previous algorithms, we assume that 
$\nor{x'}_1 \geq \alpha(1/ \delta +1)$.
We require $x'$ to be a solution with
approximation guarantee $\nor{x'}_1 \leq (1+\delta) \mathit{LIN}$ and
we require $y'$ to be an integer solution with
approximation guarantee $\nor{y'}_1 \leq (1+\delta) \mathit{LIN} +n$. For every $1\leq i \leq n$ we suppose that
$x'_i \leq y'_i$.
For a vector $z \in \mathbb{R}^{n}_{\geq 0}$, let $V(z)$ be the set of all integral vectors 
$v = (v_1, \ldots , v_n)^T$ such that $0 \leq v_i \leq z_i$.
Given LP solution $x'$ and integer solution $y'$ with the described properties above. 
The algorithm performs in the following
way.
  \begin{algo}\label{alg3}
\ 
  \begin{enumerate}
  	\item If possible choose vector $c \in V(y'-x')$ with $\nor{c}_1 = \alpha$ and return $y'' = y' -c$ and $x'' = x'$.
  	 Otherwise choose $c \in V(y'-x')$ such that $\nor{c}_1 < \alpha$ is maximal.\\
  Set $\bar{y} = y' -c$.
  \item Set $x^{var} := \frac{\alpha(1 / \delta +1)}{\nor{x'}}x'$, $x^{fix} := x' - x^{var}$
  and $b^{var} := b - A(x^{fix})$
  \item Compute an optimal solution $\hat{x}$ of the LP $\min \mengest{\nor{x}_1}{Ax \geq b^{var}, x\geq 0 }$
	\item Set $x'' = x^{fix} + \hat{x}$ 
	\item For each $1 \leq i \leq n$ set $\hat{y}_i = \max \{\lceil x''_i \rceil , \bar{y}_i \}$
	\item If possible choose $d \in V(\hat{y}-x'')$ such that $\nor{d}_1 = \alpha (1/ \delta +1)$ otherwise
  choose $d \in V(\hat{y}-x'')$ such that $\nor{d}_1 < \alpha (1/ \delta +1)$ is maximal.
  \item Return $y'' = \hat{y} -d$.
  \end{enumerate}
  \end{algo}
\begin{thm}\label{thm6}
	 Let $x'$ be a solution of the LP with $\nor{x'}_1 \leq (1+\delta) \mathit{LIN}$ and $\nor{x'}_1 \geq 
	 \alpha (1/ \delta +1)$. Let $y'$ be an integral
	 solution of the LP with  $\nor{y'}_1 \leq (1+\delta) \mathit{LIN} +n$ where $y'_i \geq x'_i$ for each 
	 $i=1, \ldots ,n$.
	Then Algorithm \ref{alg3} returns an integral solution
	$y''$ with $\nor{y''}_1 \leq (1+ \delta) \mathit{LIN} +n -\alpha$ such that
	$\sum_{y''_i<y'_i} (y'_i-y''_i) \leq \alpha (1/ \delta +2)$.
\end{thm}
\begin{proof}
	{\bf Feasibility:} Feasibility for $x''$ and approximation $\nor{x''}_1 \leq (1+ \delta)\mathit{LIN} - \alpha$ 
	follows from Theorem 3. Step 2,3 and 4 are identical to Algorithm \ref{alg1}.
	Feasibility for the integer solution $y''$ follows from the fact, that for every component $i$ we have 
	$y''_i = \hat{y}_i - d_i \geq x''_i$ and hence $Ay'' \geq Ax'' \geq b$.\\
  {\bf Size of reduction of $y'$:} The only steps where components of $y'$ are changed are in step 1, 5 and 6. 
  In step 1 we change 
	$y'$ to obtain $\bar{y}$, in step 5 we change $\bar{y}$ to obtain $\hat{y}$ and in step 6
	we change $\hat{y}$ to obtain $y''$. Summing up the change in each step leads therefore to
	the maximum possible size of reduction of $y'$ compared to $y''$.
	In step 1 there are $c \leq \alpha$ components of $y'$ which are being reduced. In step 5
	no components of $\bar{y}$ are being reduced and in step 6 there are $d \leq \alpha(1/ \delta +1)$ 
	components of $\hat{y}$ which are being reduced to obtain $y''$. Hence there are at most $\alpha(1/ \delta +2)$
	components of $y'$ which are being reduced to obtain $y''$.\\
  {\bf Approximation:} It remains to prove, that $y''$ has approximation ratio $(1+ \delta) \mathit{LIN} + n- \alpha$.\\
  Case 1, $\nor{c}_1 = \alpha$:
  In this case, the algorithm returns in step 1 solution $y'' = \bar{y}$ with $\nor{y''}_1 = \nor{y'}_1 - \alpha$
  and the algorithm terminates.
  Otherwise, if $\nor{c}_1 < \alpha$ we have for every component $i$, that $\bar{y}_i-x'_i < 1$ and $\nor{x''}_1 \leq
  (1+\delta) \mathit{LIN} - \alpha$. Note that steps 2-4 are equivalent to Algorithm \ref{alg1}.\\
  Case 2, $\nor{d}_1 < \alpha (1/ \delta +1)$:
  In this case $y''_i - x''_i = \hat{y}_i - d_i - x_i'' <1$ for $i=1, \ldots , n$, since $\nor{d}_1$
  is chosen maximally. Using $\nor{x''}_1 \leq (1+ \delta) \mathit{LIN} - \alpha$ and $y''_i < x''_i +1$ for
  $i=1, \ldots , n$ we have $\nor{y''}_1 \leq (1+ \delta) \mathit{LIN} +n - \alpha$.\\
  Case 3, $\nor{d}_1 = \alpha (1/ \delta +1)$:
  Let $\bar{m}$ be the number of components with $x_i'' > \bar{y}_i$. Next we compare the vector $\hat{y}$ with $x''$.
  Using $x'' \geq x^{fix}$ and the definition of $\hat{y}$ in step 5 we obtain 
  $\nor{\hat{y} - x''}_1 = \sum_{x_i'' \leq \bar{y}_i} (\bar{y}_i-x_i'') + \sum_{x_i'' > \bar{y}_i} 
  (\lceil x''_i \rceil - x_i'')
  \leq \sum_{x_i'' \leq \bar{y}_i}(\bar{y}_i-x_i^{fix}) + \bar{m}$.
  The fact that $\hat{y}_i-x_i'< 1$ for $i= 1, \ldots ,n$ and $\nor{x'}_1 - \nor{x^{fix}}_1  = \nor{x^{var}}_1 
  \leq \alpha(1/ \delta +1)$ and the fact there are at most $n-\bar{m}$ components with $x''_i < \bar{y}_i$ yield that
  $\sum_{x_i'' \leq \bar{y}_i}(\bar{y}_i-x_i^{fix}) = \sum_{x_i'' \leq \bar{y}_i}(\bar{y}_i- x'_i + x_i^{var}) \leq 
  n-\bar{m}+\sum_{x_i'' \leq y_i'}x_i^{var}
  \leq n-\bar{m}+\alpha(1/\delta +1)$. As a result we can bound 
  $\nor{\hat{y} - x''}_1 \leq \sum_{x_i'' \leq y_i'}(y_i'-x_i^{fix}) + \bar{m} \leq n + \alpha(1/\delta +1)$. 
  Since $y'' = \hat{y} - d$ and $\nor{d}_1= \alpha(1/\delta +1)$,
  our integer solution $y''$ has the aimed approximation guarantee of $\nor{y''}_1 =\nor{\hat{y}}_1 - \nor{d}_1 
  \leq \nor{x''} + \alpha(1/\delta +1) + n - \nor{d}_1 = \nor{x''}+n 
  \leq (1+ \delta) \mathit{LIN} +n -\alpha$.
\end{proof}
The running time of the above algorithm depends on the number of non-zero components and the time to compute 
an optimal solution of an LP. The algorithms computes an integral solution $y''$ with 
$\nor{y''}_1 \leq (1+ \delta) \mathit{LIN} +n -\alpha$ for given fractional and integral solution.
In many cases, like bin packing, the dimension $n$ is very large and provides thus a large additive
term in the approximation. The following algorithm describes how this large additive term can be avoided.
On the other hand the difference between $y'$ and $y''$ increases to $\mathcal{O}(\frac{m+\alpha}{\delta})$.
Let $x'$ be an approximate solution of the LP $\min \mengest{\nor{x}_1}{Ax \geq b, x \geq 0 }$ with
$\nor{x'}_1 \leq (1+ \delta) \mathit{LIN}$ and $\nor{x'}_1 \geq \alpha (1/ \delta +1)$. 
Furthermore let $y'$ be an approximate integer solution of the LP with $\nor{y'}_1 \leq (1+ 2\delta) \mathit{LIN}$
and $\nor{y'}_1 \geq (m+1)(1/ \delta +2)$
and $y'_i \geq x'_i$ for $i= 1, \ldots ,n$.
In addition we suppose that both $x'$ and $y'$ have exactly $K \leq \delta \mathit{LIN}$
non-zero components. 
Our goal is now to compute a fractional solution $x''$ and and integer solution $y''$ having improved approximation properties
and still $\leq \delta \mathit{LIN}$ non-zero components.
For a vector $z \in \mathbb{R}^{n}_{\geq 0}$, let $V(z)$ be the set of all integral vectors 
$v = (v_1, \ldots , v_n)^T$ such that $0 \leq v_i \leq z_i$.
Furthermore we denote with $a_1, \ldots ,a_K$ the indices of the non-zero components $y'_{a_j}$ such that
$y'_{a_1} \leq \ldots \leq y'_{a_K}$ are sorted in non-decreasing order.
\begin{algo}\label{alg4}
\ 
  \begin{enumerate}
  \item Choose $\ell$ maximally such that the sum of smallest $\ell$ components $1, \ldots , \ell$ is 
  $\sum_{1\leq i \leq \ell} y'_{a_i} \leq (m+1)(1/ \delta +2)$
	\item Set $x^{var}_i = \begin{cases} x'_i & \text{if }i= a_j \text{ for } j \leq \ell \\
	\frac{\alpha(1 / \delta +1)}{\nor{x'}}x'_i & \text{else}
	\end{cases}$ and $\bar{y}_i = \begin{cases} 0 & \text{if }i= a_j \text{ for } j \leq \ell \\
	y'_i & \text{else}
	\end{cases}$
	\item Set $x^{fix}= x' - x^{var}$, $b^{var} = b - A(x^{fix})$ and compute an 
	optimal solution $\hat{x}$ of the LP  $\min \mengest{\nor{x}_1}{Ax \geq b^{var}, x\geq 0 }$
  \item Set $x'' = x^{fix} + \hat{x}$
  \item For each $1 \leq i \leq n$ set $\hat{y}_i = \max \{\lceil x''_i \rceil , \bar{y}_i \}$
	\item If possible choose $d \in V(\hat{y}-x'')$ such that $\nor{d}_1 = \alpha (1/ \delta +1)$ otherwise
  choose $d \in V(\hat{y}-x'')$ such that $\nor{d}_1 < \alpha (1/ \delta +1)$ is maximal.
  \item Return $y'' = \hat{y} -d$
  \end{enumerate}
\end{algo}
\begin{thm}
	Let $x'$ be a solution of the LP with $\nor{x'}_1 \leq (1+\delta) \mathit{LIN}$ and 
	$\nor{x'}_1 \geq \alpha(1/ \delta +1)$. Let $y'$ be an integral
	 solution of the LP with  $\nor{y'}_1 \leq (1+2\delta) \mathit{LIN}$ and $\nor{y'}_1 \geq (m+1)(1/ \delta +2)$.
	 Solutions $x'$ and $y'$ have both exactly $K$ non-zero components and for each component we have 
	 $x'_i \leq y'_i$.
	Then Algorithm \ref{alg4} returns a fractional solution $x''$ with $\nor{x''}_1 \leq (1+ \delta) \mathit{LIN} -\alpha$
	 and an integral solution
	$y''$ with $\nor{y''}_1 \leq (1+ 2\delta) \mathit{LIN} - \alpha$. Both $x''$ and $y''$ have the same 
	number of non-zero components with $x''_i \leq y''_i$ and the number of non-zero components is bounded by
	 $\delta \mathit{LIN}$.
	The distance between $y''$ and $y'$ is bounded by $\nor{y''-y'}_1
	= \mathcal{O}(\frac{m + \alpha}{\delta})$.
\end{thm}
\begin{proof}
	{\bf Feasibility}: Feasibility and approximation for the fractional solution $x''$ follow easily from correctness 
	of Algorithm \ref{alg1} and the fact that removing additional components $x'_{a_1}, \ldots ,x'_{a_{\ell}}$ and reassigning them
	optimally does not worsen
	the approximation. Each integral component $\hat{y}_i$ is by definition (step 5) greater or equal than $x''_i$. 
	By choice of $d$ step 6 and 7 retain this property for $y''$ and imply thus feasibility for $y''$.
	
	{\bf Distance between $y''$ and $y'$:}
	The only steps where components of $y'$ are changed are step 2, 5 and 7. In step 2 we change 
	$y'$ to obtain $\bar{y}$, in step 5 we change $\bar{y}$ to obtain $\hat{y}$ and in step 7
	we change $\hat{y}$ to obtain $y''$. Summing up the change in each step leads therefore to
	the maximum possible distance between $y''$ and $y'$. In step 2 of the algorithm $\ell$ components
	of $y'$ are set to zero to obtain $\bar{y}$, which by the definition of $\ell$ results in a change of at most $(m+1)(1/ \delta +2)$.
	We define $L$ by  
	$L = \sum_{1\leq i \leq \ell} y'_{a_i}$ with $0 \leq L \leq (m+1)(1/ \delta +2)$.
	In step 5, the only components $\bar{y}_i$ being changed are the ones where $x_i''$ is larger than $\bar{y}_i$. 
	So the change in step 5 is bounded by $\sum_{x''_i> \bar{y}_i} (\lceil x''_i \rceil - \bar{y}_i) = 
	\sum_{x''_i> \bar{y}_i} (\lceil x^{fix}_i + \hat{x}_i \rceil - \bar{y}_i) 
	\leq \sum_{x''_i> \bar{y}_i}(\lceil x^{fix}_i \rceil - \bar{y}_i +\lceil \hat{x}_i \rceil) 
	\leq \sum_{x''_i> \bar{y}_i}\lceil \hat{x}_i \rceil$ by
	knowing that $\lceil x^{fix}_i \rceil - \bar{y}_i \leq 0$ since $x^{fix}_i = \bar{y}_i = 0$ if $i= a_j$ for a $j \leq \ell$
	or $\lceil x^{fix}_i \rceil < \lceil x'_i \rceil \leq y'_i$. Furthermore we can bound 
	$\sum_{x''_i> \bar{y}_i}\lceil \hat{x}_i \rceil \leq \nor{\hat{x}}_1 +m \leq \nor{x^{var}}_1 +m$ since $\hat{x}$
	is a basic feasible solution and $\nor{x^{var}}_1$
	can be bounded by $L + \alpha(1/ \delta +1)$ (i.e. we get $L$ for the size of components 
	$x'_{a_1} , \ldots, x'_{a_K}$ plus 
	$\sum_{i> \ell} \frac{\alpha(1/\delta +1)}{\nor{x'}_1} x'_{a_i} \leq \alpha(1/ \delta +1)$ for the
	remaining ones). Therefore we have $\nor{\hat{y}-\bar{y}}_1 \leq L + \alpha(1/ \delta +1) + m$.
 	In step 7, $\nor{y'' - \hat{y}}_1 = \nor{d}_1 \leq \alpha(1/ \delta +2) +m$.
	In sum this makes a total change of at most $(m+1)(1/ \delta + 2)+L + \alpha (1/ \delta + 1)+m +\alpha(1/\delta+2) +m
	\leq 2(m+1)(1/\delta +2)+2m + \alpha(2/\delta +3) = \mathcal{O}(\frac{m + \alpha}{\delta})$.

	{\bf Number of components}: 
	The property that $x'$ and $y'$ have the same number of non-zero components together with the property that 
	$y'_i \geq x'_i$ implies that $x'_i >0$ whenever $y'_i > 0$.
	This property holds also for $x^{fix}$ and $\bar{y}$ since a component $\bar{y}_i$ is set to zero if and only if
	$x^{fix}_i =0$. Notice that $y'' = \hat{y} -d \geq x''$.
	Suppose by contradiction that there is a component $i$ with $x''_i = 0$ and $y''_i >0$, then $\hat{y}_i = y''_i +d_i >0$
	and by definition
	of $\hat{y}$ we obtain $\bar{y}_i>0$. In this case we have $x^{fix}_i >0$, which gives a contradiction
	to $x''_i = 0 = x^{fix}_i + \hat{x}_i >0$.
	Using the property that $x''$ and $y''$ have the same number of non-zero components, it is sufficient to prove
	that the number of non-zero components of $x''$ is limited by $\delta \mathit{LIN}$.
	Our new solution $x''$ is composed of $x^{fix}$ and $\hat{x}$. Solution $x^{fix}$ has $K-\ell$ non-zero
	components, since in step 2 we set $\ell$ components of $x^{fix}$ to zero. Being a basic feasible solution,
	$\hat{x}$ has at most $m$ non-zero components and hence $x''$ has at most $K+m-\ell$ non-zero components.
	If $\ell \geq m$, then $x''$ has $\leq K \leq \delta \mathit{LIN}$ non-zero components.
	So let $\ell < m$: The total number of non-zero components after step 4 is $(K+m- \ell)$. We now prove that
	this number is bounded by $\delta \mathit{LIN}$. 
	Parameter $\ell$ is chosen to be maximal, therefore $\sum_{i \leq l+1} y'_{a_i} \geq (m+1)(1/ \delta +2)$. 
	Hence, the average size of
	components $y'_{a_1}, \ldots, y'_{a_{\ell +1}}$ is greater than $\frac{(m+1)(1/ \delta +2)}{\ell+1} 
	\stackrel{\ell +1 \leq m}{\geq} 
	\frac{(m+1)(1/ \delta +2)}{m}> 1/ \delta +2$. Since the components are sorted in non-decreasing order, 
	every component $y'_i$
	with $i \geq \ell +1$ has size $> 1/ \delta +2$.
	Summing over all non-zero components of $y'$ yields the following inequality:
	$\nor{y'}_1 = \sum_{i= \ell +2}^K y'_{a_i} + y'_{a_{\ell +1}} + L \geq (K-\ell -1)(1 / \delta +2)+ y'_{a_{\ell +1}} +L 
	\geq (K- \ell-1 )(1 / \delta +2)+ (m+1)(1 / \delta +2) = (K-\ell+m)(1 / \delta +2)$.
	Using that $\nor{y'}_1 \leq (1+2 \delta)\mathit{LIN}$ yields
	$(1+2 \delta)\mathit{LIN} \geq (K- \ell +m)(1 / \delta +2)$.
	Dividing both sides by $(1 / \delta +2)$ gives $(K- \ell +m) \leq \delta \mathit{LIN}$. This shows
	that the number of non-zero components of $x''$ and $y''$ is at most $\delta \mathit{LIN}$.

	{\bf Approximation:} Case1: $\nor{d}_1 = \alpha (1/ \delta +2)+ m$\\
	The following inequalities $\nor{\hat{y}}_1 \leq \nor{\bar{y}}_1 +L +\alpha(1/ \delta +2) +m = 
	\nor{y'}_1 + \alpha(1/\delta+1) +m$ and
	$\nor{y'}_1 \leq (1+2\delta) \mathit{LIN}$ together yield the aimed approximation
	$\nor{y''}_1 = \nor{\hat{y}}_1 - \nor{d}_1 = \nor{\hat{y}}_1 -\alpha(1/\delta+2) -m 
	\leq (1+2\delta) \mathit{LIN} - \alpha$.\\
	Case2: $\nor{d}_1 < \alpha (1/ \delta +2)+ m$\\
	Since $d$ is chosen maximally, $y''_i-x''_i < 1$ for every components $i= 1, \ldots ,n$. Since 
	$\nor{x''}_1 \leq (1+\delta) \mathit{LIN} - \alpha$	and $y''$ has at most $\delta \mathit{LIN}$ non-zero
	components $\nor{y''}_1$ is bounded by $(1+\delta) \mathit{LIN} - \alpha + \delta \mathit{LIN} = 
	(1+2\delta) \mathit{LIN} - \alpha$.
\end{proof}
Instead of using an optimal LP solution in Algorithm \ref{alg3} and \ref{alg4}, we can solve the LP approximately 
with a ratio of $(1+\delta /2)$. The following algorithm is basically a combination of Algorithm \ref{alg2} and Algorithm \ref{alg4}.
We could also combine Algorithm \ref{alg2} and Algorithm \ref{alg3} to obtain similar results.
We make the following assumption for the fractional solution $x'$ and the corresponding integer solution $y'$:
Let $x'$ be an approximate solution of the LP $\min \mengest{\nor{x}_1}{Ax \geq b, x \geq 0 }$ with
$\nor{x'}_1 \leq (1+ \delta) \mathit{LIN}$ and $\nor{x'}_1 \geq 2 \alpha (1/ \delta +1)$.
Let $y'$ be an approximate integer solution of the LP with $\nor{y'}_1 \leq \mathit{LIN} +2C$ for some value 
$C \geq \delta \mathit{LIN}$ and with $\nor{y'}_1 \geq (m+2)(1/\delta +2)$. 
Suppose that both $x'$ and $y'$ have only $K \leq C$
non-zero components. 
For every component $i$ we suppose that $y'_i \geq x'_i$.
Furthermore we are given indices $a_1, \ldots ,a_K$, such that the non-zero components $y'_{a_j}$ are
sorted in non-decreasing order i.e. $y'_{a_1} \leq \ldots \leq y'_{a_K}$.
\newpage
\begin{algo}\label{alg5}
\ 
  \begin{enumerate}
   \item Set $x^{var} := 2 \frac{ \alpha(1 / \delta +1)}{\nor{x'}}x'$, $x^{fix} := x' - x^{var}$ and 
   $b^{var} = b - A(x^{fix})$
	\item Compute an approximate solution $\hat{x}$ of the LP $\min \mengest{\nor{x}_1}{Ax \geq b^{var}, x\geq 0 }$
	with ratio $(1+ \delta/2)$
	\item If $\nor{x^{fix} + \hat{x}}_1 \geq \nor{x'}_1$ then set $x'' = x'$, 
	$\hat{y} = y'$ and goto step 9
  \item Choose the largest $\ell$ such that the sum of smallest components $y'_1, \ldots , y'_{\ell}$ is 
  $\sum_{1\leq i \leq \ell} y'_{a_i} \leq (m+2)(1/ \delta +2)$
	\item For all $i $ set $\bar{x}^{fix}_{i} = 
	\begin{cases} 0 & \text{if }i= a_j \text{ for } j \leq \ell \\
	x^{fix}_i & \text{else}
	\end{cases}$ 
	and $\bar{y}_i = \begin{cases} 0 & \text{if }i= a_j \text{ for } j \leq \ell \\
	y'_i & \text{else}
	\end{cases}$
	\item Set $\bar{x} = \hat{x} + x'_{\ell}$ where $x'_{\ell}$ is a vector consisting of components 
	$x_{a_1}, \ldots ,x_{a_{\ell}}$. Reduce the number of non-zero components to at most $m+1$.
  \item $x'' = \bar{x}^{fix} + \bar{x}$
  \item For all non-zero components $i$ set $\hat{y}_i = \max \{\lceil x''_i \rceil , \bar{y}_i \}$
	\item If possible choose $d \in V(\hat{y}-x'')$ such that $\nor{d}_1 = \alpha (1/ \delta +1)$ otherwise
  choose $d \in V(\hat{y}-x'')$ such that $\nor{d}_1 < \alpha (1/ \delta +1)$ is maximal.
  \item Return $y'' = \hat{y} -d$
  \end{enumerate}
\end{algo}
Step 6 of the algorithm can be performed using a standard technique presented for example in \cite{beling1998}. 
Arbitrary many components of 
$\bar{x}$ can be reduced to $m+1$ without making the approximation guarantee worse.
We formulate the following theorem and corollary such that we can directly use it in the next section.
\begin{thm}\label{thm8}
	Let $x'$ be a solution of the LP with $\nor{x'}_1 \leq (1+\delta) \mathit{LIN}$ and $\nor{x'}_1 \geq 
	2 \alpha (1/ \delta +1)$. Let $y'$ be an integral
	 solution of the LP with $\nor{y'}_1 \leq \mathit{LIN} +2C$ for some value $C \geq \delta \mathit{LIN}$
	 and with $\nor{y'}_1 \geq (m+2)(1/\delta +2)$.
	 Solutions $x'$ and $y'$ have the same number of non-zero components and for each component we have 
	 $x'_i \leq y'_i$. The number of non-zero components of $x'$ and $y'$ is $K$ with $K \leq C$.
	Then Algorithm \ref{alg5} returns a fractional solution $x''$ with $\nor{x''}_1 \leq (1+ \delta) \mathit{LIN} -\alpha$
	 and an integral solution
	$y''$ where one of the two properties hold:
	 $\nor{y''}_1 = \nor{y'}_1 - \alpha$ or $\nor{y''}_1 = \nor{x''}_1 + C$. 
	 Both, $x''$ and $y''$ have at most $C$
	non-zero components and the distance between $y''$ and $y'$ is bounded by $\nor{y''-y'}_1 
	= \mathcal{O}(\frac{m + \alpha}{\delta})$.
\end{thm}
\begin{proof}
	Note that the first 3 steps are equivalent to Algorithm \ref{alg2}. In steps 4-6 
	the number of non-zero components $x'_{a_1}, \ldots ,x'_{a_{\ell}}$ are reduced. As we apply
	a method that does not increase the objective value we obtain by
	Theorem 4 that $\nor{x''}_1 \leq (1+\delta) \mathit{LIN} - \alpha$. 
	Steps 4-9 are similar to Algorithm \ref{alg4}. The main
	difference is that components $x'_{a_1}, \ldots ,x'_{a_{\ell}}$ are not assigned by the LP but are added to the 
	LP solution	afterwards in step 7.\\
	{\bf Distance between $y''$ and $y'$:} As in Theorem 7, the steps where components of $y'$ are changed are steps 
	5,8 and 10. By definition of $\ell$ the change of $y'$ in step 5 is bounded by $(m+2)(1/ \delta +2)$.
	As shown, the change in step 8 is bounded by $\sum_{x''_i> \bar{y}_i} \lceil x^{fix}_i + \hat{x}_i \rceil - \bar{y}_i) 
	\leq \sum_{x''_i> \bar{y}_i}\lceil \bar{x}_i \rceil \leq \nor{\lceil \bar{x}_i \rceil}_1$ and 
	$\nor{\lceil \bar{x}_i \rceil}_1 \leq 2 \alpha (1/ \delta +1) + L+1$, where$L = \sum_{1\leq i \leq \ell} y'_{a_i}$.
	The change in step 10 is bounded by $\nor{d}_1 \leq 2 \alpha (1/ \delta +2)+ m +1$. Therefore the total change
	between $y'$ and $y''$ is bounded by $\mathcal{O}(\frac{m + \alpha}{\delta})$.\\
	{\bf Number of components:} According to Theorem 7, the number of nonzero components of $y''$ is equal to 
	the number of non-zero components of $x''$ which equals $K - \ell + m+1$ (the number of non-zero components of
	$\hat{x}$ is bounded by $m+1$). We distinguish between the
	two cases where $\ell \geq m+1$ and $\ell < m+1$. In the case where $\ell \geq m+1$ the number of components
	of $x''$ is smaller than $K$ and hence bounded by $C$. Consider the case where $\ell < m+1$.
	By definition of $\ell$ we know that $L + y'_{\ell +1}
	\geq (m+2)(1/ \delta +2)$.
	Using the argument in the proof of Theorem 7, we obtain the following inequality:
	$\nor{y'}_1 = \sum_{i= \ell +2}^k y'_i + y'_{\ell +1} + L = (K-\ell -1)(1 / \delta +2)+ y'_{\ell +1} +L 
	\geq (K- \ell-1 )(1 / \delta +2)+ (m+2)(1 / \delta +2) = (K-\ell+m +1)(1 / \delta +2)$
	Using that $\nor{y'}_1 \leq \mathit{LIN} + 2C$ yields
	$\mathit{LIN} + 2C \geq (K- \ell +m +1)(1 / \delta +2)$.
	As $\frac{\mathit{LIN} +2C}{(1 / \delta +2)} =\frac{ \delta\mathit{LIN} + 2 \delta C}{(1+2\delta)}
	\stackrel{C \geq \delta\mathit{LIN}}{\leq} \frac{C + 2 \delta C}{(1+ 2 \delta)} = C$ we obtain that 
	$(K- \ell +m +1) \leq \frac{\mathit{LIN} + 2C}{(1 / \delta +2)} \leq C$.\\
	{\bf Approximation:} According to Theorem 7 we distinguish between the two cases where 
	$\nor{d}_1 = 2 \alpha (1/ \delta +2) +m+1$ and $\nor{d}_1 < 2 \alpha (1/ \delta +2) +m+1$. In the second case where
	$\nor{d}_1 < 2 \alpha (1/ \delta +2) +m+1$ we know that $\nor{y''}_1$ is bounded by $\nor{x''}_1$ plus the number
	of non-zero components of $x''$ since $d$ is chosen maximally. Hence $\nor{y''}_1 \leq \nor{x''}_1 + C$.
	In the case where $\nor{d}_1 = 2 \alpha (1/ \delta +2) +m+1$, we know $\nor{y''}_1 \leq \nor{\hat{y}}_1 -
	2 \alpha (1/ \delta +2) -m-1$. As $\nor{\hat{y}}_1 \leq \nor{\bar{y}}_1 + \nor{\lceil \bar{x} \rceil}_1 \leq 
	\nor{y'}_1 + 2\alpha(1/\delta+1) +m+1$ we get $\nor{y''}_1 \leq \nor{y'}_1 - \alpha$.
	Note that we can also make the general claim for $y''$ that $\nor{y''}_1 \leq \mathit{LIN} +2C - \alpha$. 
\end{proof}
The following corollary is an analog to Corollary \ref{cor5} which states what Algorithm \ref{alg5} is doing if the approximation
ratio of $x'$ is worse than $(1+ \delta)$. We will need this corollary in the next section as we have no true control
about the approximation ratio of $x'$. During the bin packing algorithm new columns might appear in the LP, 
which might change
the optimal solution and therefore the approximation ratio of a solution $x'$.
\begin{cor}\label{cor9}
	Let $\nor{x'}_1 = (1+ \delta')\mathit{LIN}$ for some $\delta' \geq \delta$ 
	 and $\nor{x'}_1 \geq 2 \alpha (1/ \delta +1)$
	and let $\nor{y'}_1 \leq \mathit{LIN} + 2C$ for some
	$C \geq \delta'\mathit{LIN}$ and $\nor{y'}_1 \geq (m+2)(1/\delta +2)$. 
	Solutions $x'$ and $y'$ have the same number of non-zero components and for each component we have 
	 $x'_i \leq y'_i$. The number of non-zero components of $x'$ and $y'$ is $K$ with $K \leq C$. 
	 Then Algorithm \ref{alg5} returns a fractional solution
	$x''$ with $\nor{x''}_1 \leq \nor{x'}_1 - \alpha = (1+ \delta')\mathit{LIN} - \alpha$ and an integral solution $y''$ 
	where one of the two properties holds:
	 $\nor{y''}_1 = \nor{y'}_1 - \alpha$ or $\nor{y''}_1 = \nor{x'}_1 - \alpha + C$. 
	Both $x''$ and $y''$ have at most $C$
	non-zero components and the distance between $y''$ and $y'$ is bounded by $\nor{y''-y'}_1 
	= \mathcal{O}(\frac{m + \alpha}{\delta})$.
\end{cor}
\begin{proof}
Note that steps 1-3 are basically identical to Algorithm \ref{alg2}. Hence Algorithm \ref{alg5} returns by Corollary 6 a fractional
solution $x''$ with $\nor{x''}_1 \leq \nor{x'}_1 - \alpha$. The distance between the integral solutions $y'$ 
and $y''$ are independent of the approximation ratio of $x'$. Hence the distance between $y'$ and $y''$ is
according to Theorem \ref{thm8} bounded by $\mathcal{O}(\frac{m + \alpha}{\delta})$. The number of non-zero components of
$x''$ and $y''$
is by the proof of Theorem \ref{thm8} bounded by the number $K \leq C$ of non-zero components of $y'$ or by 
$\frac{\mathit{LIN} +2C}{1/ \delta +2} \leq C$.
The approximation guarantee for $y''$, that $\nor{y''}_1 \leq \nor{y'}_1 - \alpha$ follows if 
$\nor{d}_1 = 2 \alpha (1/ \delta +2) +m+1$. If $\nor{d}_1 < 2 \alpha (1/ \delta +2) +m+1$ then $\nor{y''}_1 \leq 
\nor{x''}_1 + C \leq \nor{x'}_1 + C - \alpha$.
We can also make the general claim for $y''$ that $\nor{y''}_1 \leq \nor{y'}_1 - \alpha$.
\end{proof}

\section{AFPTAS for robust bin packing}
\label{sec:bin-packing}
The goal of this section is to give a fully robust AFPTAS for the bin packing
problem using the methods developed in the previous section. For that purpose we show at first the common way how one
can formulate a rounded instance of bin packing as an ILP. In Section 4.2 we present abstract properties
of a rounding that need to be fulfilled to obtain a suitable rounding and in Section 4.3 we present the used
dynamic rounding algorithm. The crucial part however is the analysis of the dynamic rounding in combination with ILP
techniques. Since the ILP and its optimal value are in constant change due to the dynamic rounding, it is difficult to 
to give a bound for the approximation. Based on the abstract properties we therefore develop techniques how
to view and analyze the problem as a whole. 

The \emph{online bin packing problem} is defined as follows:
Let $I_t = \{i_1, \ldots i_t\}$ be an instance with $t$ items at time step $t \in \mathbb{N}$
and let $s:I_t \to (0,1]$ be a mapping that defines the sizes of the items.
Our objective is to find a function $B_{t}:\menge{i_1,\ldots,i_t}\to
\mathbb{N}^+$, such that $\sum_{i:B_{t}(i)=j}s(i)\leq 1$ for all
$j$ and minimal $\max_{i}\menge{B_{t}(i)}$ (i.e. $B_t$ describes a packing of the items into a minimum number
of bins). We allow to move
few items when creating a new solution $B_{t+1}$ for instance $I_{t+1} = I_t \cup \{i_{t+1}\}$. 
Sanders et al. \cite{sanders2009} and also Epstein and Levin \cite{epstein2006robust}
defined the \emph{migration factor} to give a measure for the amount of repacking. 
The migration factor is defined as the total size of all items that are
moved between the solutions divided by the size of the arriving item.
Formally the migration factor of two packings $B_t$ and $B_{t+1}$ is defined by 
$\sum_{j \leq t: B_t(i_j) \neq B_{t+1}(i_j)} s(i_j)/s(i_{t+1})$.
\subsection{LP-Formulation}
Let $I$ be an instance of bin packing with $m$ different item sizes $s_1, \ldots, s_m$. Suppose
that for each item $i_k \in I$ there is a size $s_j$ with $s(i_k) = s_j$.
A configuration $C_i$ is a multiset of sizes $\{ a(C_i,1):s_1,  a(C_i,2):s_2, \ldots a(C_i,m):s_m \}$
with $\sum_{1\leq j\leq m} a(C_i,j)s_j \leq 1$, where $a(C_i,j)$
denotes how often size $s_j$ appears in configuration $C_i$. We denote by $C$ the set of all configurations. Let $|C|=n$.
We consider the following LP relaxation of the bin packing problem:
\begin{align*}
  &\min \nor{x}_1\\
  &\sum_{C_i \in C}x_{i} a(C_i,j) \geq b_j \qquad \forall 1 \leq j \leq m\\
  &x_{i}\geq 0 \qquad \forall 1 \leq j \leq n
\end{align*}
Component $b_j$ states the number of items $i$ in $I$ with $s(i) = s_j$ for $j = 1, \ldots , m$.
This LP-formulation was first described by Eisemann
\cite{eisemann1957trim}. Suppose that each size $s_j$ is larger or equal to
$\epsilon/2$ for some $\epsilon \in (0,1/2]$. Since the number of different item sizes is $m$,
the number of feasible packings for a bin is bounded by $|C|= n \leq
(\frac{2}{\epsilon} +1)^{m}$.
Obviously an optimal integral solution of the LP gives a solution to our bin packing problem. We denote
by $\mathit{OPT}(I)$ the value of an optimal solution. An optimal fractional solution is a lower bound 
for the optimal value. We denote the optimal fractional solution by $\mathit{LIN}(I)$.


\subsection{Rounding}
We use a rounding technique based on the offline APTAS by Fernandez de La Vega \&
Lueker \cite{de1981bin}. As we plan to modify the rounding through the dynamic rounding algorithm we give a 
more abstract approach on how we can
round the items to obtain an approximate packing.
At first we divide the set of items into \emph{small} ones and
\emph{large} ones. An item $i$ is called \emph{small} if $s(i) <
\epsilon/2$, otherwise it is called \emph{large}. Instance $I$ is partitioned accordingly
into the large items $I_{L}$ and the small items $I_{S}$.
We treat small items and large items differently. Small items can be packed using a greedy algorithm
and large items need to be rounded using a rounding function.
We define a \emph{rounding function} as a function $R: I_L \mapsto \mathbb{N}$ which maps each large item $i$
to a \emph{group} $j$. By $R^j$ we denote the set of items being mapped to the same 
group $j$, i.e.
$R^j = \{i \in I_L \mid R(i)=j \}$. By $\lambda_{j}^R$ we denote an item $i$ with 
$s(i) = \max \{s(i_k) \mid i_k \in R^j\}$.
Given an instance $I$ and a rounding function $R$, we define the rounded instance $I^R$ by rounding the size of 
every large item $i \in R^j$ for $j \geq 1$ up to the size $s(\lambda_{j}^R)$ of the largest item in its group. 
Items in $R^0$ are
excluded from instance $I^R$. We write $s_R (i)$ for the rounded size of item $i$ in $I^R$.
Depending on constants $c$ and $d$, we define the following properties for a rounding function $R$.
\begin{itemize}
\item[(A)] $\max\menge{R(i) \mid i \in I_L} = c/\epsilon^{2}$ for a constant $c \in \mathbb{R}^+$
\item[(B)] $|R^i| = |R^j|$ for all $i,j \geq 2$
\item[(C)] $|R^0| = d |R^1|$ for a constant $d \in \mathbb{R}^+$ with $d \geq 1$
\item[(D)] $s(i)\leq s(j)\Leftrightarrow R(i)\geq R(j)$ 
\end{itemize}
Any rounding function fulfilling property (A) has at most $\Theta(1/ \epsilon^2)$ different item sizes and hence 
instance
$I^R$ can now be solved approximately using the LP relaxation. The resulting LP relaxation has $\Theta(1/ \epsilon^2)$ 
rows
and can be solved approximately with accuracy $(1+ \delta)$ using the max-min resource sharing \cite{grigoriadis2001approximate}
in polynomial time.
Based on the fractional solution we obtain an integral solution $y$ of the LP with $\nor{y}_1 \leq (1+\delta)\mathit{LIN}(I^R) + C$
for some additive term $C \geq 0$.
We say a packing $B$ \emph{corresponds} to a rounding $R$ and solution $y$ if items 
in $R^1, \ldots ,R^m$ are packed by $B$ according to the integral solution $y$ of the LP. The LP is defined by
instance $I^R$. Items in $R^0$ are each packed in separate bins.
\begin{lem}\label{lem10}
  Given instance $I$ with items greater than $\epsilon/2$ and a rounding function $R$ fulfilling properties
  (A) to (D), then $\mathit{OPT}(I^R) \leq \mathit{OPT}(I)$ and $|R^0| \leq \frac{2d}{c} \epsilon \mathit{OPT}(I)$.
  Let $y$ be an integral solution of the LP for instance $I^R$ with $\nor{y}_1 \leq (1+ \delta)\mathit{LIN(I^R)} +C$ 
  for some value $C \geq 0$, let $B$ be a packing of $I$ which corresponds to $R$ and $y$ and let  
  $\epsilon' =\frac{2d}{c} \epsilon$. Then
  \begin{align*}
    \max_{i}\menge{B_{t}(i)} = \nor{y}_1 + |R^0| \leq (1+\epsilon'+\delta)\mathit{OPT}(I)+C.
  \end{align*}
\end{lem}
\begin{proof}
	Let $m = \max\menge{R(i) \mid i \in I_L}$. Let $R^i$ be the set of items in rounding group $i$, which corresponds
	to their rounded sizes and let $\mathcal{R}^i$ be the set of items in $R^i$, which corresponds to their actual
	size. Instance $I^R$ contains every item from $R^1$ to $R^m$, 
	while items from $R^0$ are excluded.
	By property (D) we know, that items in $R^i$ are larger or equal than items in $R^{i+1}$.
	By property (C) we find for every item in $\mathcal{R}^1$ an unique item in $R^0$ with larger or equal size, 
	since the largest item
	in $R^0$ to which all items are being rounded up is smaller than any item in $R^1$.
	Using property (B) for each item in $\mathcal{R}^{i+1}$ we find a unique larger item in $R^i$.
	Therefore we have for every item in the rounded instance $I^R$ an item with larger size in instance $I$ and hence
	\begin{align*}
	\mathit{OPT(I^R)} \leq \mathit{OPT(I)}.
   \end{align*}
   Since the packing $B$ corresponds to a solution $y$, $B$ gives a solution with 
   $\max_{i}\menge{B_{t}(i)} \leq (1+\delta)\mathit{LIN}(I^R)+C+|R^0|$ bins and since $\mathit{LIN}(I^R) \leq \mathit{OPT(I^R)} \leq \mathit{OPT(I)}$
   we obtain that $\max_{i}\menge{B_{t}(i)} \leq  (1+\delta)\mathit{OPT}(I)+C+|R^0|$. Further, we can bound $|R^0|$.
   Since every item in $I$ is of size at least $\epsilon/2$ there is a lower bound for the optimum:
   $\mathit{OPT(I)} \geq \epsilon/2 \sum_{0 \leq i \leq m} |R^i| \geq \epsilon/2 \sum_{0 \leq i \leq m} |R^0|/d
   = \frac{\epsilon (m+1) |R^0|}{2 d} \geq \frac{c |R^0|}{2 d \epsilon}$.
   Resolving this inequality, we get $|R^0| \leq \frac{2 \epsilon d }{c} \mathit{OPT}(I)$ and hence 
	$|R^0| \leq \epsilon' \mathit{OPT}(I)$. Since
   $c$ and $d$ are constant we know $|R^0| \leq \epsilon' \mathit{OPT}(I) = \mathcal{O}(\epsilon \mathit{OPT(I)})$ which implies together with
   the inequality $\max_{i}\menge{B_{t}(i)} \leq (1+\delta)\mathit{OPT}(I)+C+|R^0| \leq 
   (1+\epsilon'+\delta)\mathit{OPT}(I)+C$.
\end{proof}
How can we handle the small items? Actually, small items do not make problems at all. We can pack them via 
FirstFit \cite{coffman1984approximation}
on top of the existing large items and still obtain a good solution. FirstFit is a greedy algorithm which simply places
the current item into the first bin having enough space. FirstFit opens a new bin of the item does not fit into any
used bin.
\begin{lem}\label{lem11} \cite{de1981bin}
	Let $I$ be an instance with small and large items and given a packing $B$ of the large items with
	$\max_{i}\menge{B_{t}(i)} \leq K$ for some $K \geq 1$.
  Packing the small items via FirstFit on top of packing $B$ gives a new packing of instance $I$ which uses
  \begin{align*}
    \max \menge{K,(1+\epsilon)\mathit{OPT}(I)+1}
  \end{align*}
  bins.
\end{lem}
Given instance $I = \{i_1, \ldots ,i_t \}$, we define $m$ by $m = \lceil 1/\epsilon^2 \rceil$ if 
$\lceil 1/\epsilon^2 \rceil$ is even and otherwise $m = \lceil 1/\epsilon^2 \rceil +1$. By definition $m$ is
always even.
For every instance $I$ we find a rounding function $R$
with rounding groups $R^0, R^1, \ldots R^m$ which fulfills properties (A)-(D) such that $|R^0| < 2|R^1|$
and $|R^0| \geq |R^1|$.
\begin{algo}\label{alg6}
\   
\begin{enumerate}
\item Partition the large items according to the rounding function $R$ in groups $R^0, \ldots, R^m$
\item Round up the size of each large item $i \in R^1, \ldots ,R^m$ to $s(\lambda_{i}^R)$ to obtain instance $I^R$
\item Compute a fractional solution $x$ of the LP defined by $I^R$ approximately with ratio $(1+\bar{\delta})$
\item Round up each component of the fractional solution to obtain an integral solution $y$ for the LP for instance $I^R$
\item Pack items in $R^1, \ldots ,R^m$ according to the integral solution $y$
\item Open a bin for each item $i$ with $R(i)=0$
\item Pack the small items in $I_{S}$ via FirstFit
\end{enumerate}
\end{algo}
A solution $x$ of $I^R$ with ratio $(1+\bar{\delta})$ having $m+1$ non-zero components can be computed using
max-min resource sharing \cite{grigoriadis2001approximate}.
According to Lemma \ref{lem10} and 11, the algorithm described above produces a solution with approximation
$\leq (1+ \epsilon' + \bar{\delta})\mathit{OPT}+m+1$ with $\epsilon' \leq \frac{2d}{c} \epsilon
\leq 4 \epsilon$.
\subsection{Online Bin Packing}
Let us consider the case where items arrive online. As new items arrive we are allowed to repack several items but
we intend to keep the migration factor as small as possible.
We present operations that modify the current rounding $R_t$ and packing $B_t$ to give a solution for
the new instance. The given operations worsen the approximation but by applying the results from the previous section
we can maintain an approximation ratio that depends on $\epsilon$.
The presented rounding technique is similar to the one used in \cite{epstein2006robust}. 
In our algorithm we use approximate solutions of ILPs in contrast to the APTAS of Epstein \& Levin which solve
the ILPs optimally. Handling with approximate ILPs results in a different analysis of the algorithm because
many helpful properties of optimal solution are getting lost.

Note that in an online scenario of bin packing  where large and small items arrive online, small items do not
need to be considered. We use the same techniques as in \cite{epstein2006robust} to pack small items. As a small item 
arrives we place it via FirstFit \cite{coffman1984approximation}. 
In this case FirstFit 
increases the number of bins being used by at most $1$ (\cite{de1981bin}) and the migration factor is zero as we 
repack no item.
Whenever a new large item arrives several small items might also need to be replaced. Every small item in a bin
that is repacked by the algorithm, is replaced via FirstFit. Packing small items with this strategy 
does not increase the number of bins that need to
be repacked as a large item arrives. Later on the migration factor will solely be determined by the number of bins
that are being repacked. More precisely, we will prove that the number of bins, that need to be repacked is bounded
by $\mathcal{O}(1/\epsilon^3)$.
Therefore we assume without loss of generality that every arriving item is large, i.e. has a size $\geq \epsilon/2$ 
(see also \cite{epstein2006robust}).
Our rounding $R_{t}$ will be constructed by three different \emph{operations},
called the \emph{insertion, creation} and \emph{union} operation. The
insertion operation is performed whenever a large item arrives. This
operation is followed by a creation or an union operation depending on the phase the algorithm is in. 
Let $I = \{i_1, \ldots ,i_t \}$ be the existing instance as defined above, let $R$ be the corresponding rounding function, let $x$ 
be a fractional solution of the LP generated for the rounded instance $I^R$ and let $B$ be the current packing
of items in $I$. We define two subgroups of $R^0$ denoted by $R^{1.5}$ and $R^{2.5}$ in the creation phase, 
which are also
being modified by the operations.
Let $I' = I \cup \{ i_{t+1} \}$ be the new instance. We use the following operations that modify
the current rounding $R$, the packing $B$ and the fractional and integral LP solution $x$ and $y$.
We denote with $R'$, $B'$, $x'$ and $y'$ the new rounding, packing and fractional/integral LP solutions for instance $I'$
\subsubsection*{Insertion Step}
Find the largest $j$ with $s(\lambda_{j}^R) \geq s(i_{t+1})$. 
Set $R'(i_{t+1})=j$ and $B'(i_{t+1})= B(\lambda_{j}^R)$.
For every $k= 1, \ldots, j$ we define $R'(\lambda_{k}^R) = k-1$ and $B'(\lambda_{k}^R) = B(\lambda_{k-1}^R)$.
Set $x'=x$ and $y'=y$.
\begin{figure}
\includegraphics[width=\textwidth]{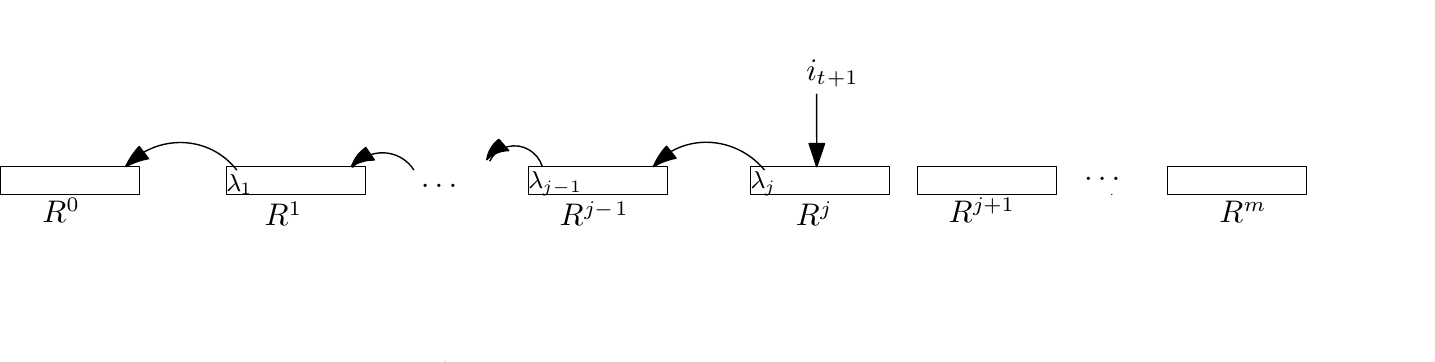}
\caption{Insert operation}
\end{figure}
\subsubsection*{Modified Insertion Step}
During the creation phase, the algorithm uses the modified insertion operation.
Find the largest $j$ ($j=1.5$ and $j= 2.5$ included) with $s(\lambda_{j}^R) \geq s(i_{t+1})$. 
Set $R'(i_{t+1})=j$ and $B'(i_{t+1})= B(\lambda_{j}^R)$.
For every $k= 1,4,5, \ldots, j$ we define $R'(\lambda_{k}^R) = k-1$ and $B'(\lambda_{k}^R) = B(\lambda_{k-1}^R)$.
For every $k = 1.5,2,2.5,3$ we define $R'(\lambda_{k}^R) = k-0.5$ and $B'(\lambda_{k}^R) = B(\lambda_{k-0.5}^R)$.
Set $x'=x$ and $y'=y$.

\subsubsection*{Creation Phase}
The creation phase consists of $k$ creation steps, where $k= |R^1|$. At the end of each creation phase we intend
to have new rounding groups $R^1$ and $R^2$ created from the subgroups of $R^0$ named $R^{1.5}$ and $R^{2.5}$.
At the beginning of the creation phase we always have $|R^0|= 2k$ and $R^{1.5}$ and $R^{2.5}$ are empty.
In the first step we change the rounding group
for all items $i$ with $R(i)=j \geq 1$ to $R'(i)=j+2$. 
Furthermore we say the $k$ largest items of $R^0$ belong to $R^{1.5}$ and the $k$ smallest items belong to $R^{2.5}$.
In each of the $k$ creation steps we change the rounding function for the largest items $\lambda_{1.5}^R$ 
and $\lambda_{1.5}^R$. Set $R'(\lambda_{1.5}^R) = 1$
and $R'(\lambda_{2.5}^R) = 2$. Since items $\lambda_{1.5}^R$ and $\lambda_{2.5}^R$ are moved from $R^0$ to $R^1$ 
and $R^2$ they have to be covered by
the LP. Therefore we increase the value of the LP solution by $x'_i = x_i + 1$, $x'_j = x_j+1$ and 
$y'_i = y_i + 1$, $y'_j = y_j+1$, where $i,j$ are defined such that 
$C_i = \{1:s_{R'}(\lambda_{1.5}^R) \}$ and $C_j = \{1:s_{R'}(\lambda_{2.5}^R) \}$. For $k \neq i,j$ set $x'_k = x_k$
and $y'_k = y_k$.
\begin{figure}
\includegraphics[width=\textwidth]{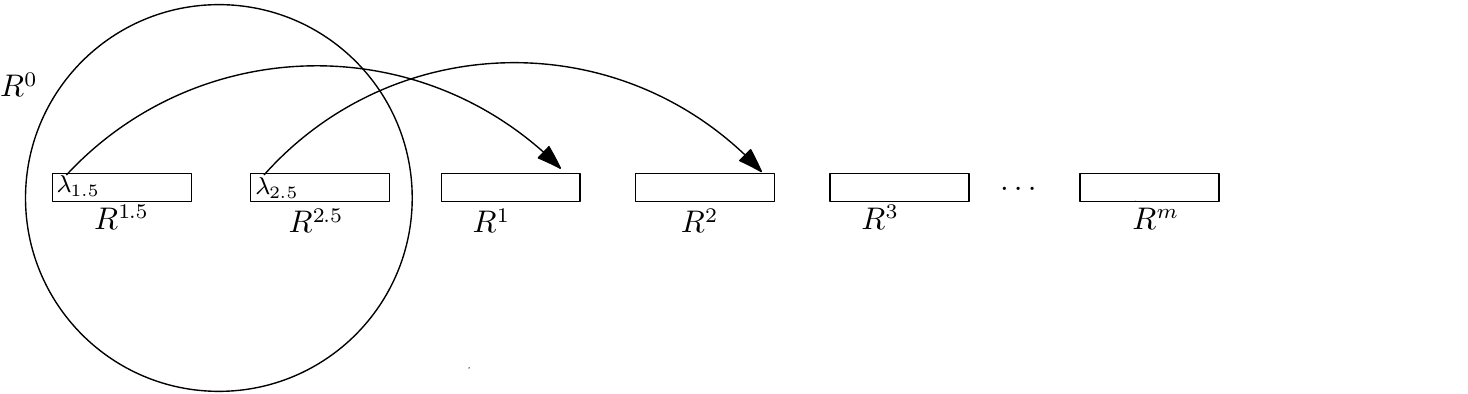}
\caption{Create operation}
\end{figure}
\subsubsection*{Union Phase}
The union phase consists of $k$ union steps, where $k= |R^1|$. At the end of each union phase we have made
out of $4$ roundings groups $2$ rounding groups with size doubled.
For the first union step we determine the largest index $j$ with $|R^j| < |R^{j+1}|$. If there is no such index then
set $j=m$. In each step now set $R'(\lambda_{j}^R) = j-1$ and $R'(\lambda_{j-2}^R) = j-3$ and for the other items $i$
we define $R'(i) = R(i)$. Modify the packing for $\lambda_{j}^R$ and $\lambda_{j-2}^R$ by $B'(\lambda_{j}^R) = 
B(\lambda_{j-2}^R)$ and
place $\lambda_{j-2}^R$ into a new bin. Modifying the packing this way implies that we have to change one configuration
of the fractional and integral LP solution $x$ and $y$ and add one configurations for the additional bin. 
Let $C_i$ be the
configuration used by $B'(\lambda_j)$. Configuration $C_i$ is replaced by a configuration $\hat{C_i}$ where an item
of size $s_{R'}(\lambda_{j-2}^R)$ is exchanged by an item of size $s_{R'}(\lambda_{j}^R)$. Furthermore we add another 
configuration $C_{\ell}$ with an item of size $s_{R'}(\lambda_{j-2}^R)$.
\begin{figure}
\includegraphics[width=\textwidth]{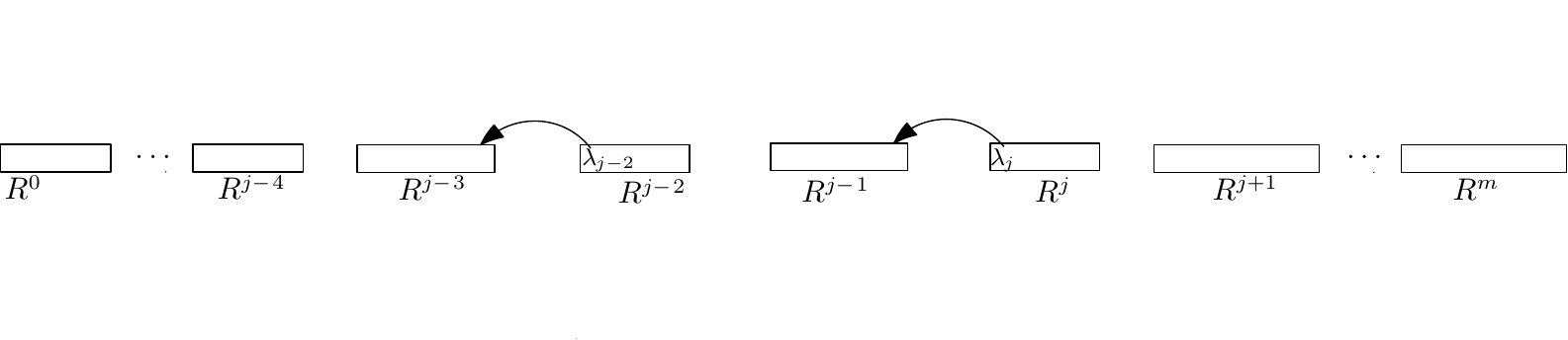}
\caption{Union operation}
\end{figure}

Note that each repacking that we perform in the operations is valid because we always replace items by smaller ones.
New packings $B'$ are created in a way that they correspond to new integer solution $y'$. We have to prove that
this solutions $y'$ is feasible. Note also, that in a creation operation and in a union operation two additional
non-zero components of size $1$ might be created.
\begin{lem}\label{lem12}
	Applying any operation above on a rounding $R$ and ILP solution $y$ with corresponding packing $B$ defines
	a new rounding $R'$ and a new integral solution $y'$. Solution $y'$ is a feasible solution of the LP
	for instance $I^{R'}$.
\end{lem}
\begin{proof}
	We have to analyze how the LP for instance $I^{R'}$ changes in comparison to the LP for instance $I^R$.
	\\{\bf Insertion Operation:} The right hand side of the LP derived from $R'$ does not change at all since
	the right hand side is determined by the cardinalities $|R'^1|=|R^1|, \ldots, |R'^m|=|R^m|$. 
	For some $j \geq 1$ let $R^j$ be the the rounding group where the new item is inserted.
	By construction of the insertion operation for each rounding group $R^{\ell}$ with $\ell = 1, \ldots, j$, there is
	one item that is inserted into group $R'^{\ell}$ and one item that is shifted out.
	Let $\iota_{\ell}^R$ be the second largest item of rounding group $R^{\ell}$.
	Since the largest item $\lambda_{\ell}^R$ in group $R^{\ell}$ is shifted to the next group,
	the size $s_{R'}(i)$ of item $i$ in a group $R^{\ell}$ is defined by $s_{R'}(i) = \iota_{\ell}^R$.
	Therefore each item in $I^{R'}$ is rounded to the previous smaller value since $s(\iota_{\ell}^R) 
	\leq s(\lambda_{\ell}^R$).
	Hence configurations of the LP solution for $I^R$ can be transformed into feasible configurations for $I^{R'}$
	i.e. $\nor{y'}_1 = \nor{y}_1$.\\
	{\bf Creation Operation:}
	Note that the rounding groups $R^{\ell}$, for $\ell = 1, \ldots, m+2$ remain identical; i.e. $R'^{\ell} = R^{\ell}$.
	The groups $R'^1$ and $R'^2$ get both a new item, but of smaller size. Therefore the sizes $s_r(i)$ of all items
	$i \in I^R$ are not modified by a creation operation. We have $s_R(i) = s_{R'}(i)$ for items in groups 
	$R^1, \ldots , R^{m+2}$. Therefore the matrix $A = (a((i,j)))$ remains the same. Only the right hand side $b'$
	of the LP from instance $I^{R'}$ is modified (i.e. $\nor{b'-b}_1 =2$).
	As two new configurations are being added to $x$ and $y$ they cover exactly the enhanced right hand side
	and are therefore a feasible solution of the LP from instance $I^{R'}$.\\
	{\bf Union Operation:}
	In the union operation we basically change only 4 rounding groups. Suppose we merge rounding group $R^{j-3}$ with
	$R^{j-2}$ and rounding group $R^{j-1}$ with $R^{j}$. While the size of $|R'^{j-3}| = |R^{j-3}|+1$ and 
	$|R'^{j-1}| = |R^{j-1}|+1$ is incremented the size of $|R'^{j}| = |R^{j}|-1$ and $|R'^{j-2}| = |R^{j-2}|-1$
	is reduced. Similar to the creation operation, this leads to a change in the right
	hand side of the LP. Two components of the right hand side, which correspond to $s(\lambda_{j}^{R'})$
	and $s(\lambda_{j-2}^{R'})$ are reduced by $1$ and two other components, which correspond to the 
	$s(\lambda_{j-1}^{R'})$ and $s(\lambda_{j-3}^{R'})$ are increased by $1$.
	Furthermore the sizes of items in $R'^{j}$ and $R'^{j-2}$ are equal or smaller than the sizes of
	items in $R^j$ and $R^{j-2}$ since $s_{R'}(i) = \iota_{j}^R$
	for all items $i \in R'^{j}$ and $s_{R'}(i) = \iota_{j-2}^R$ for all items $i \in R'^{j-2}$.
	 $\lambda_{j}^R$
	and $\lambda_{j-2}^{R}$ are shifted to the next rounding groups. 
	Consider a feasible configuration $C$ of the LP for instance $I^R$. Then the modified configuration $\bar{C}$
	(with replaced item sizes) is also feasible in the LP for instance $I^{R'}$. The new solutions $x'$ and $y'$ use
	the modified configurations and cover the right hand side of the LP.
\end{proof}
The operations are used as described in Algorithm \ref{alg7} below. We apply the algorithm on a rounding function
$R_0$ and instance $I_0 = \{i_1, \ldots i_{T} \}$. We suppose that $|R_{0}^0| = |R_{0}^1| = \ldots = |R_{0}^m| = K$
for some $K >0$ and hence $T = K(m+1)$.
An improve($a,x,y,\bar{\delta}$) statement stands for a call of Algorithm \ref{alg5} 
with improvement $\alpha= a$, fractional solution $x$ and integral solution $y$. The variable part $x^{var}$ 
is defined by $x^{var} = 2 
\frac{\alpha(1/ \bar{\delta} +1)}{\nor{x}_1}x$. 
After an improve call the packing is changed according to the new integral solution.
Since during a creation operation and a union operation two additional non-zero components of size $1$ might appear,
we change the parameter $\ell$ of Algorithm \ref{alg5} slightly to $\ell'$. Parameter $\ell'$ is defined maximally such that
the sum of the smallest components $y'_1, \ldots , y'_{\ell}$ are 
$\sum_{1\leq i \leq \ell} y'_{a_i} \leq (m+2)(1/ \delta +2) +2$. The two additional non-zero components belong to
components $y_1 , \ldots y_{\ell +2}$ and are therefore reduced in step 6 along with the others.
\begin{algo}\label{alg7}
$\phantom{f}$\\
 \begin{algorithm}[H]
    \For{i := 1 to K}{
	get new item\;
	improve($1, x,y,\bar{\delta}$);
       insert\;}
	\For{i := 1 to m/2}{
		\tcc{Creation Phase}
		\For{j := 1 to K}{
			get new item\;
			improve($1, x,y,\bar{\delta}$)\;
       		 modified insert\;
			create\;
		}
		\tcc{Union Phase}
		\For{j := 1 to K}{
			get new item\;
			improve($2, x,y,\bar{\delta}$)\;
       		insert\;
			union\;
		}
	}
  \end{algorithm}
\end{algo}
In the following we present how the algorithm changes the rounding groups for $m = 6$. The table presents the state of each
rounding
groups after each phase.
\begin{table}\centering
\begin{tabular}{c|c|c|c|c|c|c|c|c|c}
phases & $|R^0|$ & $|R^1|$ & $|R^2|$ & $|R^3|$ & $|R^4|$ & $|R^5|$ & $|R^6|$ & $|R^7|$ & $|R^8|$\\
\hline
start & $K$ & $K$ & $K$ & $K$ & $K$ & $K$ & $K$ & $0$ & $0$ \\
insertion & $2K$ & $K$ & $K$ & $K$ & $K$ & $K$ & $K$ & $0$ & $0$ \\
creation & $K$ & $K$ & $K$ & $K$ & $K$ & $K$ & $K$ & $K$ & $K$\\
union & $2K$ & $K$ & $K$ & $K$ & $K$ & $2K$ & $2K$ & $0$ & $0$ \\
creation & $K$ & $K$ & $K$ & $K$ & $K$ & $K$ & $K$ & $2K$ & $2K$ \\
union & $2K$ & $K$ & $K$ & $2K$ & $2K$ & $2K$ & $2K$ & $0$ & $0$ \\
creation & $K$ & $K$ & $K$ & $K$ & $K$ & $2K$ & $2K$ & $2K$ & $2K$ \\
union & $2K$ & $2K$ & $2K$ & $2K$ & $2K$ & $2K$ & $2K$ & $0$ & $0$ \\
\end{tabular}
\end{table}
One can see that after the execution of Algorithm \ref{alg7} each rounding group has exactly $2K$ items. We prove the
general case for arbitrary $m$: Every rounding has exactly $2K$ items after the execution of Algorithm \ref{alg7}.
\begin{lem}\label{lem13}
	Let $R_0$ be the rounding function at the beginning of the algorithm. Suppose that every 
	rounding group	$R_{0}^0, \ldots, R_{0}^m$	has exactly $K$ items.
	Then after the execution of the algorithm above the computed rounding function $R^T$ after $T$ insertions
	has $m+1$ rounding groups $R_{T}^0, \ldots,R_{T}^m$ with $|R_{T}^{\ell}| = 2K$ for $\ell = 0, \ldots, m$. 
\end{lem}
\begin{proof}
	The algorithm starts with a rounding function that contains exactly $T$ items. After the first $K$ insertion steps
	rounding function $R_K$ is of the form:
	$|R_{K}^0|=2K, |R_{K}^1|=K, \ldots , |R_{K}^m|=K$ since $K$ items are shifted to $R^0$  while the cardinalities 
	of the other rounding groups remain the same.
	During the next $K$ arrivals, the algorithm is in the creation phase. We perform a creation operation after
	each insertion. For each item shifted to $R^0$, two items are shifted to the new created groups $R^1$ and $R^2$.
	At the end of the first creation
	phase, the rounding function $R_{2K}$ satisfies $|R_{2K}^0|=K, |R_{2K}^1|=K, \ldots , |R_{2K}^{m+2}|=K$.
	In the following union phase, the rounding groups $R^{m+2} $, $R^{m+1} $ and $R^{m}$, $R^{m-1}$
	are merged together. For each union operation, one item is shifted from $R^{m+2} $ to $R^{m+1}$
	and another from $R^{m} $ to $R^{m-1} $. Since there are $K$ insert operations in the union phase,
	rounding group $|R_{3K}^0|=2K, |R_{3K}^1|=K, \ldots , |R_{K}^{m-2}|=K, 
	|R_{K}^{m-1}|=2K, |R_{K}^{m}|=2K$. 
	After the next creation and union phase, the number of rounding groups is also $m+1$.
	On the other hand we have two additional groups $|R_{5K}^{m-3}| = |R_{5K}^{m-2}|	= 2K$.
	After $j<m/2$ creation and union phases the rounding
	function $R_{2jK+K}$ is by induction of the form $|R_{2jK+K}^0|=2K, |R_{2jK+K}^1|=K, \ldots , |R_{2jK+K}^{m-2j}|=K, 
	|R_{2jK+K}^{m-2j+1}|=2K, \ldots, |R_{2jK+K}^{m}|=2K$. 
	This can be proved by induction on $j$. For $j = m/2 -1$ we get $|R_{2mK-K}^0|=2K, |R_{2mK-K}^1|=K, 
	|R_{2mK-K}^2| = K, |R_{2mK-K}^3| = 2K, \ldots |R_{2mK-K}^m| = 2K$.
	Therefore, after one additional creation and union phase, we obtain $m+1$ groups $|R_{2mK+K}^{\ell}|=2K$ for
	$\ell = 0, \ldots , m+1$.
\end{proof}
Using Algorithm \ref{alg7} with a starting rounding function $R_T$ that has $m$ rounding groups and the property that 
$|R_{T}^0|= |R_{T}^1|= \ldots = |R_{T}^m|$
produces according to the lemma a rounding function $R_{2T}$ that has also $m$ rounding groups of equal size,
but with cardinality doubled.
Therefore we can use Algorithm \ref{alg7} repetitively to always get suitable rounding functions.
The following algorithm is our final online AFPTAS for the classical bin packing problem. Let $S_t$ be the sum
of all item sizes of items $i_1, \ldots i_t$.
\begin{algo}\label{alg8}
\
\begin{itemize}
\item While $S_t \leq (m+2)(1/ \bar{\delta} +4)$ and $(m+1)$ does not divide $t$ get the new item $i_{t+1}$ and
use the offline AFPTAS \ref{alg6} with an LP of approximation ratio $(1+\bar{\delta})$.
\item Afterwards use Algorithm \ref{alg7} repetitively to obtain a packing for each instance
\end{itemize}
\
\end{algo}
By using the offline AFPTAS for small instances we can make sure that Algorithm \ref{alg7} is started with
a suitable rounding function. Since Algorithm \ref{alg7} always produces a rounding function fulfilling properties (A) to (D)
and $m+1$ divides the current number of items $t$, every rounding group $R^0, \ldots , R^m$ has the same
number of item sizes as the algorithm leaves the while-loop in the first step.

In the following we give a bound for the rounding functions $R_t$ that we produce in every step of
the algorithm. 
It remains to prove that the approximation during the execution of Algorithm \ref{alg7} can be bounded.
Therefore we define a relation between rounding functions.
Let $R$ and $\bar{R}$ be two rounding functions, with $\bar{R}$ having $\bar{m}$ rounding groups for 
some $\bar{m} \in \mathcal{O}(1/\epsilon^2)$.
We can \emph{embed} $R$ into $\bar{R}$ in symbols $R \leq \bar{R}$, if $|R^0| \leq |\bar{R}^0|$ and for 
every item $i \in I \setminus \bar{R}^0$ we 
have $s_{R}(i) \leq s_{\bar{R}}(i)$. A relation $R \leq \bar{R}$ always implies that 
$|R^0| + \mathit{OPT}(I^{R}) \leq |\bar{R}^0| + \mathit{OPT}(I^{\bar{R}})$.
\begin{lem}\label{lem14}
  For each $t\in \mathbb{N^{+}}$, we can embed $R_{t}$ into a function
  $\bar{R}_t$, which fulfills properties (A) to (D). Rounding function $\bar{R}_t$ has parameter
  $c \geq 1/4$ for property (A) and $d \leq 2$ for property (C).
\end{lem}
\begin{proof}
	Since we basically shift largest items to the following rounding group we designed operations 
	insertion, creation and union in a way that property (D) is never being
	violated by any $R_j$ for $j \leq t$. As shown in the proof of Lemma \ref{lem12} the number of rounding groups remains
	constant between $m+1$ at the end of a union phase and $m+3$ during the creation and union phase.
	Suppose the algorithm above has started with a number of items $T<t \leq 2T$ and rounding groups 
	$R_{T}^0, \ldots R_{T}^m$,which are being modified	by the algorithm.
	We define the rounding function $R = \bar{R}_t$ in which $R_t$ can be embeded in the following way: 
	Function $\bar{R}$ has rounding groups $\bar{R}^0, \ldots , \bar{R}^{\lfloor \frac{t}{2K}\rfloor -1}$ with 
	$|\bar{R}^1|= \ldots = |\bar{R}^{\lfloor \frac{t}{2K}\rfloor-1}| = 2K$ and $|\bar{R}^0| = 2K + (t \mod 2K)$.
	Since every rounding group of $\bar{R}$ except $\bar{R}^0$ has the same number of items, the rounding function
	$\bar{R}$ fulfills property (B). Rounding function $\bar{R}$ fulfills property (C) because $2K \leq |\bar{R}^0| 
	\leq 4K$ and $|\bar{R}^1| = 2K$. This implies constant $d \leq 2$.
	We prove property (A) by giving an upper and a lower bound for $\max_{i}\menge{\bar{R}(i)}$ that are both
	in $\Theta(\frac{1}{\epsilon^2})$. Recall that $T/K = m+1$.
	On the one hand we get $\max_{i}\menge{\bar{R}(i)} =\lfloor \frac{t}{2K}\rfloor -1 \leq 
	\frac{t}{2K} -1 \leq \frac{2T}{2K} -1 \leq (m+1)-1 = m \leq \lceil \frac{1}{\epsilon^2} \rceil+1
	\leq \frac{1}{\epsilon^2} +2$.
	On the other hand $\max_{i}\menge{\bar{R}(i)} =\lfloor \frac{t}{2K}\rfloor -1 \geq \lfloor \frac{T}{2K}\rfloor -1
	\geq \lfloor \frac{m+1}{2} \rfloor -1 \geq \frac{m}{2} -1 = \frac{1}{2 \epsilon ^2} -1$. Since $\epsilon \leq 1/2$
	weg get $c \geq 1/4$.
	It remains to prove that we can embed $R_t$ in $\bar{R}$ i.e. $R_t \leq \bar{R}$. 
	Since $R^0$ never exceeds $2K$ items and $2K \leq 4K$ we get $|R_{t}^0| \leq |\bar{R}^0|$.
	According to the proof of Lemma \ref{lem12} and the construction of the creation operation, $R_t$ is during 
	the creation phase
	of the following form: $|R_{t}^0| = 2K-a, |R_{t}^1| = a, |R_{t}^2| = a, |R_{t}^3| = \ldots = |R_{t}^j|= K,
	|R_{t}^{j+1}| = \ldots |R_{t}^{m+2}| = 2K$ for some $a \leq K$ and $j \leq m+2$. Rounding function $\bar{R}$
	has in every rounding group $\bar{R}^j$ for $j\geq 1$ exactly $2K$ items.
	Since property (D) 
	holds for both rounding function $R_t$ and $\bar{R}$, the rounding groups $R_{t}^{j+1}, \ldots ,R_{t}^{m+2}$ contain
	the same items as the last $m+2 - j$ rounding groups of $\bar{R}$.
	Items in these groups are therefore rounded identically.
	For some $\bar{m}$ let $\bar{R}^{\bar{m}}$ be the rounding group which contains the same items as $R_{t}^{j+1}$.
	Since rounding groups $R_{t}^3, \ldots ,R_{t}^j$ each contain exactly $K$ items, the rounding groups 
	$\bar{R}^1, \ldots ,\bar{R}^{\bar{m}- 1}$ contain the items of exactly two rounding groups. Therefore the items
	in $\bar{R}^1, \ldots ,\bar{R}^{\bar{m}- 1}$ are rounded to a smaller size compared to using
	$\bar{R}$. Items that belong to $R_{t}^1$ and $R_{t}^2$
	are contained in $\bar{R}^0$ and by definition do not need to be considered. Hence, any rounding function $R_t$
	which is in an creation phase can be embedded into an $\bar{R}$.
	By construction of the union operation and the proof of Lemma \ref{lem12}, $R_t$ is during the union phase
	of the form $|R_{t}^0| = K+a, |R_{t}^1| = \ldots = |R_{t}^{j-4}|= K,|R_{t}^{j-3}|= K+a, |R_{t}^{j-2}|= K-a,
	|R_{t}^{j-1}|= K+a, |R_{t}^{j}|= K-a, |R_{t}^{j+1}| = \ldots |R_{t}^{m+2}| = 2K$ for some $a \leq K$ and 
	$j \leq m+2$. As shown in the union phase, items in $R_{t}^{j+1},\ldots ,R_{t}^{m+2}$ are rounded equally
	in $\bar{R}$. As the sum of $R_{t}^{j-1}$ and $R_{t}^{j}$ is $2K$ and the sum of $R_{t}^{j-3}$ and $R_{t}^{j-2}$
	is $2K$ the items of $R_{t}^{j-1}$ and $R_{t}^{j}$ and the items in $R_{t}^{j-3}$ and $R_{t}^{j-2}$ 
	belong in $\bar{R}$ to the same rounding group and are hence rounded equally or to a smaller size compared to using
	$\bar{R}$.
	Items in $R_{t}^1,\ldots ,R_{t}^{j-4}$ are each of size $K$ and are rounded equally or to a smaller size 
	than using $\bar{R}$ since the same argument as in the creation phase holds.
\end{proof}
Define $\bar{\epsilon}$ by $\bar{\epsilon} = \frac{1}{16} \epsilon$.
As $\bar{R}_t$ fulfills property (A) to (D), we obtain by Lemma \ref{lem10} and $R_t \leq \bar{R}_t$ the following 
two equations for every $t$:
\begin{enumerate}
\item $\mathit{OPT}(I^{\bar{R}_t}) \leq \mathit{OPT}(I_t)$
\item $|R_{t}^0| \leq |\bar{R}_{t}^0| \leq \frac{2d}{c} \epsilon \mathit{OPT}(I_t) \leq 
\bar{\epsilon}\mathit{OPT}(I_t)$
\end{enumerate}
Note that since $\epsilon \leq 1/2$ we have $\bar{\epsilon} \leq \frac{1}{32}$. Recall that
$R_t \leq \bar{R}_t$ implies that 
$|R^{0}_t| + \mathit{OPT}(I^{R_t}) \leq |\bar{R}^{0}_t| + \mathit{OPT}(I^{\bar{R_t}})$ and that 
$\mathit{LIN}(I^{R_t}) +m \geq \mathit{OPT}(I^{R_t})$ (rounding up a basic feasible solution).
Let us discuss how the methods from the previous section apply to the presented online algorithm.
The procedure improve is implemented by using Algorithm \ref{alg5} in order 
to get an improved solution for instance $I^{R_t}$. 
Algorithm \ref{alg5} is applied using $\bar{\delta}$ as the approximation parameter. 
In the following lemma we prove that
applying Algorithm \ref{alg5} to improve a solution for $I^{R_t}$ impacts the overall approximation 
$\Delta = \bar{\epsilon} + \bar{\delta} + \bar{\epsilon} \bar{\delta}$
in the same way. We define $C = \Delta\mathit{OPT}(I_t)+m$.
\begin{thm}\label{thm15}
	Given a rounding function $R_t$ and an LP defined for $I^{R_t}$. Let $x$ be a fractional solution of the LP with
	$\nor{x}_1 + |R_{t}^0|  \leq (1+ \Delta) \mathit{OPT}(I_t)$ and $\nor{x}_1 \geq 2\alpha(1/\bar{\delta} +1)$ 
	and $\nor{x}_1 = (1+\delta')\mathit{LIN}(I^{R_t})$ for some 
	$\delta'>0$. Let $y$ be	an integral solution of the LP with $\nor{y}_1 \geq (m+2)(1/\bar{\delta} +2)$ and
	corresponding packing $B$ such that $\max_i B_t (i) = \nor{y}_1 + |R_{t}^0| \leq (1+ 2\Delta) \mathit{OPT}(I_t)+m$.
	Suppose $x$ and $y$ have the same number $\leq C$ of non-zero components and for all components $i$ we have
	$y_i \geq x_i$.
	Then using 
	Algorithm \ref{alg5} on $x$ and $y$ returns new solutions $x'$ with $\nor{x'}_1 + |R_{t}^0| \leq (1+ \Delta) 
	\mathit{OPT}(I_t) -
	\alpha$ and integral solution $y'$ with corresponding packing $B'_t$ such that
	$$\max_i B'_t (i) \leq (1+ 2 \Delta) \mathit{OPT}(I_t) +m- \alpha.$$
	Further, both solutions $x'$ and $y'$ have the same number $\leq C$ of non-zero components and for each component we have
	$x'_i \leq y'_i$.
\end{thm}
\begin{proof}
	As shown in the following, Algorithm \ref{alg5} maintains the property that $x$ and $y$ have
	the same number of non-zero components and that $x_i \leq y_i$ since we can use Theorem \ref{thm8} and Corollary \ref{cor9}.
	By condition we have $\max_i B_t (i) = \nor{y}_1 + |R_{t}^0| \leq (1+ 2\Delta) \mathit{OPT}(I_t)+m$. Since
	$ \mathit{OPT}(I_t) \leq  \mathit{OPT}(I^{R}_t) + |R_{t}^0|$ we obtain for the integral solution $y$ that
	$\nor{y}_1 \leq 2\Delta \mathit{OPT}(I_t)+m + \mathit{OPT}(I^{R}_t) \leq  2\Delta \mathit{OPT}(I_t)+m + \mathit{LIN}(I^{R}_t) +m$.
	Hence by definition of $C$ we get $\nor{y}_1 \leq  \mathit{LIN}(I^{R}_t) + 2C$. This is one requirement to use Theorem \ref{thm8} 
	or Corollary \ref{cor9}.
	We look at the cases separately where on the one hand $\delta' \leq \bar{\delta}$ and on the other hand 
	$\delta' > \bar{\delta}$. 
	
	Case 1, $\delta' \leq \bar{\delta}$:
	At first we give an upper bound for $\mathit{LIN(I^{R_t})}$: We get $\mathit{LIN(I^{R_t})} \leq \mathit{OPT}(I^{R_t}) \leq 
	\mathit{OPT}(I^{R_t}) +  |R_{t}^0| \leq
	\mathit{OPT}(I^{\bar{R_t}}) + |\bar{R_{t}}^0| \leq (1+ \bar{\epsilon}) \mathit{OPT}(I_t)$ using
	that $R_t \leq \bar{R_{t}}$. This implies
	that $\bar{\delta} \mathit{LIN(I^{R_t})} \leq \bar{\delta} \mathit{OPT}(I^{R_t}) \leq
	 (\bar{\delta} + \bar{\delta}\bar{\epsilon}) \mathit{OPT}(I_t) < C$.
	Algorithm \ref{alg5} returns by Theorem \ref{thm8} a solution $x'$ with $\nor{x'}_1 \leq 
	(1+\bar{\delta})\mathit{LIN(I^{R_t})}-\alpha$ and an integral solution $y'$ with 
	$\nor{y'}_1 \leq \nor{x'}_1 + C$ or $\nor{y'}_1 \leq \nor{y}_1 - \alpha$.
	For the term $\nor{x'}_1 + |R_{t}^0|$ we get  $\nor{x'}_1 + |R_{t}^0| \leq 
	(1+ \bar{\delta}) \mathit{OPT}(I^{R_t}) - \alpha + |R_{t}^0|$.
	Using that $R_t$ can be embedded in $\bar{R_t}$ we get $|R_{t}^0| + \mathit{OPT}(I^{R_t}) \leq |\bar{R_t}^0| + 
	\mathit{OPT}(I^{\bar{R_t}}) \leq \mathit{OPT}(I_t) + \bar{\epsilon} \mathit{OPT}(I_t)$. 
	Therefore $\nor{x'}_1 + |R_{t}^0| \leq \bar{\delta} \mathit{OPT}(I^{R_t})
	- \alpha + \mathit{OPT}(I_t) + \bar{\epsilon} \mathit{OPT}(I_t) \leq 
	(\bar{\delta} + \bar{\delta}\bar{\epsilon}) \mathit{OPT}(I_t) - \alpha + (1+ \bar{\epsilon}) \mathit{OPT}(I_t)
	\leq (1 + \Delta) \mathit{OPT}(I_t) - \alpha$.
	In the case where $\nor{y'}_1 \leq \nor{x'}_1 + C$ we can bound the number of bins of the new packing $B'$ 
	by	$\max_i B'_t (i) = \nor{y'}_1 + |R_{t}^0| \leq \nor{x'}_1+ |R_{t}^0| + C  
	\leq (1 + \Delta) \mathit{OPT}(I_t) - \alpha + C = (1 + 2 \Delta) \mathit{OPT}(I_t) +m - \alpha$.
	In the case that $\nor{y'}_1 \leq \nor{y}_1 - \alpha$ we obtain
	$\max_i B'_t (i) = \nor{y'}_1 + |R_{t}^0| \leq \nor{y}_1 - \alpha + |R_{t}^0| =  \max_i B_t (i) - \alpha
	\leq (1+ 2\Delta) \mathit{OPT}(I_t) +m- \alpha$.
	
	Case 2, $\delta' > \bar{\delta}$: 
	By condition we have $\nor{x}_1 + |R_{t}^0| \leq (1+ \Delta) \mathit{OPT}(I_t)$. Since
	$ \mathit{OPT}(I_t) \leq  \mathit{OPT}(I^{R_t}) + |R_{t}^0|$ we obtain for the solution $x$ that
	$\nor{x}_1 \leq \Delta \mathit{OPT}(I_t) + \mathit{OPT}(I^{R_t}) \leq  \Delta \mathit{OPT}(I_t) + \mathit{LIN}(I^{R_t}) +m$.
	Hence by definition of $C$ this implies $\nor{x}_1 \leq  \mathit{LIN}(I^{R_t}) + C$ and therefore $\delta' \mathit{LIN}(I^{R_t}) < C$,
	which fulfills the requirements of Corollary \ref{cor9}.
	Using Algorithm \ref{alg5} on solutions $x$ with $\nor{x}_1 = (1+\delta')\mathit{LIN}(I^{R_t})$ and $y$ with 
	$\nor{y}_1 \leq  \mathit{LIN}(I^{R_t}) + 2C$
	we obtain by Corollary \ref{cor9} a fractional solution $x'$ with
	$\nor{x'}_1 \leq \nor{x}_1 - \alpha$ and an integral solution $y'$ with either
	$\nor{y'}_1 \leq \nor{y}_1 - \alpha$ or $\nor{y'}_1 \leq \nor{x}_1 + C - \alpha$.
	So for the new packing $B'$ we can guarantee, that
	$\max_i B'_t (i) = \nor{y'}_1 + |R_{t}^0| \leq \nor{y}_1 - \alpha + |R_{t}^0| =  \max_i B_t (i) - \alpha
	\leq (1+ 2\Delta) \mathit{OPT}(I_t) +m - \alpha$
	if $\nor{y'}_1 \leq \nor{y}_1 - \alpha$. If $\nor{y'}_1 \leq \nor{x}_1 + C - \alpha$, we can guarantee
	that $\max_i B'_t (i) = \nor{y'}_1 + |R_{t}^0| \leq \nor{x}_1 + |R_{t}^0| + C - \alpha \leq 
	(1+ \Delta) \mathit{OPT}(I_t)+ C - \alpha \leq (1+ 2\Delta) \mathit{OPT}(I_t) +m 
	- \alpha$.
	Furthermore we know by Corollary \ref{cor9} that $x'$ and $y'$ have at most $C$ 
	non-zero components.
\end{proof}
Set $\bar{\delta} = \bar{\epsilon}$. Then $\Delta = 2 \bar{\epsilon} + \bar{\epsilon}^2 = \mathcal{O}(\epsilon)$.
We get the central theorem:
\begin{thm}\label{thm16}
  Algorithm \ref{alg8} is a fully robust AFPTAS for the bin packing problem.
\end{thm}
\begin{proof}
	While instances are small Algorithm \ref{alg8} uses the offline AFPTAS (see Algorithm \ref{alg6}). Using Algorithm \ref{alg6}, we get a packing
	$B_t$ for instance $I_t$ that uses at most
	$\max_i B_t(i) \leq (1+ \epsilon' + \bar{\delta})\mathit{OPT}(I_t)+ \frac{1}{\epsilon^2}+1$ bins, where 
	$\epsilon' \leq 4 \epsilon < \bar{\epsilon}$. 
	Since the instance is small
	the migration factor is bounded although we might repack every single item. Let $\tau$ be the first index where
	the algorithm leaves the while-loop. By condition we are in the while loop
	while $S_t \leq (m+2)(1/ \bar{\delta} +4)$ and $t$ does not divide $m+1$. Hence $S_{\tau} \leq 
	(m+2)(1/ \bar{\delta} +4) +m = \mathcal{O}(1/\epsilon^3)$. The migration factor for instances $I_t$ with 
	$t \leq \tau$ is therefore bounded by $\frac{2}{\epsilon} S_t = \mathcal{O}(1/\epsilon^4)$ since every arriving item
	has size at least $\epsilon/2$. The approximation guarantee for small instances is bounded by 
	$\max_i B_t(i) \leq (1+\bar{\delta} + \bar{\epsilon}) \mathit{OPT}(I_t) +m +1$.
	In the following we consider large instances $I_t$ with $t \geq \tau$.
	
	{\bf Full robustness:} The migration factor for some consecutive packings $B_t$ and $B_{t+1}$ is bounded by
	the migration of the improve-call plus the migration of an insertion and an union operation. 
	The operations create requires no shifting of items at all.
	As proven in the previous section, an improve-call changes at most $\mathcal{O}(m/ \bar{\delta})$ components of a solution $y$.
	Since the arriving item is large with size $\geq \epsilon/2$, changing a complete configuration 
	requires migration of at most $\mathcal{O}(1 /\epsilon)$.
	Combined this results in a migration factor for the improve-call $\mathcal{O}(m/\Delta^2) = \mathcal{O}(1/\epsilon^4)$ 
	if we use Algorithm \ref{alg5}.
	By construction of the insertion operation it shifts in worst case one item per 
	rounding group. Having $\mathcal{O}(1/ \epsilon^2)$ rounding groups this gives a migration factor of at most 
	$\mathcal{O}(1/ \epsilon^3)$.
	Therefore the complete migration is bounded by $\mathcal{O}(1/\epsilon^4)$.
	
	{\bf Running time:} The running time is dominated by the max-min resource sharing (see Algorithm \ref{alg5}) and 
	the number of non-zero components. The number of non-zero components is bounded by $\Delta \mathit{OPT}(I_t) +m
	\leq \Delta t + \frac{1}{\epsilon^2} +1$
	and is therefore polynomial in $\frac{1}{\epsilon}$ and $t$. As the running time for the max-min resource sharing
	is also polynomial in $\frac{1}{\epsilon}$ (see \cite{grigoriadis2001approximate}), 
	the running time is clearly polynomial in $t$ and $\frac{1}{\epsilon}$.
	
	{\bf Approximation:} We prove by induction that four properties hold for any packing $B_t$
	and corresponding LP solutions. Given fractional solutions
	$x$ and integral solution $y$ of the LP defined by instance $I^{R_t}$. 
	Properties (2)-(4) are
	necessary to apply Theorem \ref{thm16} and property (1) provides the wished approximation ratio for the bin packing problem.
	\begin{enumerate}
		\item packing $B_t$ uses at most $(1+ 2\Delta)\mathit{OPT}(I_t) +m$ bins
		\item $\nor{x}_1 + |R_{t}^0|  \leq (1+ \Delta) \mathit{OPT}(I_t)$
		\item for every configuration $i$ we have $x_i \leq y_i$
		\item $x$ and $y$ have the same number of non-zero components and that number is bounded by 
		$\Delta \mathit{OPT}(I_t) +m$
	\end{enumerate}
	To apply Theorem \ref{thm15} we furthermore need a guaranteed minimal size for $\nor{x}_1$ and $\nor{y}_1$.
	According to Theorem \ref{thm15} integral solution $y$ needs $\nor{y}_1 \geq (m+2)(1/ \bar{\delta} +2)$ and 
	$\nor{x}_1 \geq 4 (1/ \bar{\delta} +1)$ as we set at most $\alpha = 2$.
	By condition of the while-loop we know that any instance $S_t \geq (m+2)(1/ \bar{\delta} +6)$. 
	Since $\mathit{OPT}(I_t) \leq \nor{y}_1 + |R_{t}^0| \leq \nor{y}_1 + \bar{\epsilon} \mathit{OPT}(I_t)$ we get
	$\nor{y}_1 \geq (1-\bar{\epsilon})\mathit{OPT}(I_t) = (1-\bar{\delta})\mathit{OPT}(I_t)$. 
	By $\mathit{OPT}(I_t) \geq (m+2)(1/ \bar{\delta} +4)$ we
	finally get that $\nor{y}_1 \geq (1-\bar{\delta}) (m+2)(1/ \bar{\delta} +6) \geq 
	(m+2)(1/ \bar{\delta} +6) - (m+2)(1 + 6\bar{\delta}) \geq (m+2)(1/ \bar{\delta} +6) - 4(m+2) = 
	(m+2)(1/ \bar{\delta} +2)$. Since $\mathit{OPT}(I_t) \leq \nor{x}_1 + m + |R_{t}^0|$ we obtain by the same argument
	that $\nor{x}_1 \geq (m+2)(1/ \bar{\delta} +2) -m \geq (m+2)(1/ \bar{\delta} +1)$ and since $m = 1/\epsilon 
	\stackrel{\epsilon \leq \bar{\delta}}{\geq} 1/ \bar{\delta}
	\geq 2$ we get that $\nor{x}_1 \geq 4 (1/ \bar{\delta} +1)$.
	
	In the case that $t = \tau$ we have by the offline algorithm that the number of non-zero components $= m+1
	\leq \Delta \mathit{OPT}(I_t)+m$ since $\mathit{OPT}(I_t) \geq S_t > 1/ \Delta$. The number of used bins is bounded
	by $\max_i B_t (i) < (1+ \bar{\delta} + \bar{\epsilon}) \mathit{OPT}(I_t) +m +1 < (1+ 2\Delta)\mathit{OPT}(I_t) +m$
	(note $\epsilon' < \bar{\epsilon}$) and property (2) is fulfilled for the same reason. Furthermore in the 
	offline algorithm every component $x_i$ is rounded up to obtain the integral component $y_i$. 
	Therefore all properties (1)-(4) are fulfilled for $t \leq \tau$ and the induction basis holds.
	Now let $B_t$ be a packing for $t > \tau$ for instance $I_t$ with solutions $x$ and $y$ of the LP defined by 
	$I^{R_t}$. Suppose
	by induction that property (1)-(4) hold. We have to prove that these properties also hold for $B_{t+1}$ and the
	corresponding solutions of the LP defined by $I^{R_{t+1}}$. Packing $B_{t+1}$ is created by using an
	improve call for $x$ and $y$ followed by an insertion operation and optional, an union or a creation operation.
	\\{\bf improve:} Let $x'$ be the resulting fractional solution of Algorithm \ref{alg5}, let $y'$ be the resulting integral solution
	of Algorithm \ref{alg5} and let $B'_t$ be the corresponding packing. Properties (1)-(4) are fulfilled 
	for $x$, $y$ and $B_t$ by induction hypothesis. Hence we can use Theorem \ref{thm15}. 
	By Theorem \ref{thm15} properties (1)-(4) are then still fulfilled for $x'$, $y'$ and $B'_t$ and moreover we get
	$\nor{x'}_1 + |R_{t}^0| \leq (1+ \Delta) \mathit{OPT}(I_t)- \alpha$ and 
	$\max_i B'_t (i) \leq (1+ 2 \Delta) \mathit{OPT}(I_t) + m - \alpha$ for $\alpha = 2$ or $\alpha = 2$.
	\\{\bf operations:} First we take a look at how the operations modify $\nor{x'}_1$, $\nor{y'}_1$ and $|R_{t}^0|$.
	By construction of the insertion operation, the LP solutions $x'$ and $y'$ are not modified while $|R_{t}^0|$
	increases by $1$. By construction of the creation operation $\nor{x'}_1$ and $\nor{y'}_1$ are increased by $2$
	and $|R_{t}^0|$ decreases by $2$. By construction of the union operation, $\nor{x'}_1$ and $\nor{y'}_1$ 
	are increased by $1$ and $|R_{t}^0|$ remains constant.
	Property (1): Let $x''$ be the fractional solution and $y''$ be the integral solution after using operations 
	on $x'$ and $y'$. Packing $B_{t+1}$ equals $\max_i B_{t+1} = \nor{y''}_1 + |R^{0}_{t+1}|$. According to the operations
	an insertion operation yields $\max_i B_{t+1} = \nor{y'}_1 + |R^{0}_{t}| +1 = \max_i B'_t +1$.
	An insertion operation followed by an union operation yields $\max_i B_{t+1} = \nor{y'}_1 + 1+ |R^{0}_{t}| +1= 
	\max_i B'_t +2$ and an insertion operation followed by a creation operation
	yields $\max_i B_{t+1} = \nor{y'}_1 + 2 +|R^{0}_{t}| -1 = \max_i B'_t +1$. Algorithm \ref{alg7} is designed that in 
	the union phase $\max_i B'_t \leq (1+ 2 \Delta) \mathit{OPT}(I_t) + m - 2$ since there is an improve call with 
	$\alpha =2$ and otherwise
	$\max_i B'_t \leq (1+ 2 \Delta) \mathit{OPT}(I_t) + m - 1$ since there is an improve call with $\alpha = 1$.
	Therefore we have in any case that $B_{t+1}$ uses at most $(1+ 2\Delta)\mathit{OPT}(I_{t}) +m
	\leq (1+ 2\Delta)\mathit{OPT}(I_{t+1}) +m $ bins.
	The proof that property (2) holds is symmetric since $\nor{x'}_1$ increases in the same way as $\nor{y'}_1$ 
	and $\nor{x'}_1 + |R_{t}^0|  \leq (1+ \Delta) \mathit{OPT}(I_t) - \alpha$ for $\alpha=1$ or $\alpha=2$.
	For property (3) note that in the operations a configuration $x_i$ of the fractional solution is increased by $1$ 
	if and only
	if a configuration $y_i$ is increased by $1$. Therefore the property that for all configurations $x''_i \leq y''_i$ 
	retains from $x'$ and $y'$. By Theorem \ref{thm15} the number of non-zero components of $x'$ and $y'$ is bounded 
	by $\Delta \mathit{OPT}(I_t) +m \leq \Delta \mathit{OPT}(I_{t+1}) +m$. 
	By construction of the creation operation and union operation $x''$ and $y''$
	might have two additional non-zero components. But since these are being reduced by Algorithm \ref{alg5} (note that 
	we increased the number of components being reduced in step 6 by $2$), the LP solutions $x''$ and $y''$ 
	have at most $\Delta \mathit{OPT}(I_{t+1}) +m$ non-zero components which proves property (4).
	\end{proof}
\subsection{Running Time}
Storing items that are in the same rounding group in a heap structure, we can perform each operation 
(insertion, creation and union) in time $\mathcal{O}(\frac{1}{\epsilon^2}\log(\epsilon^2 t))$. 
Furthermore Algorithm \ref{alg5} needs to look through all non-zero components. The number of non-zero components is bounded
by $\mathcal{O}(\epsilon \mathit{OPT}) = \mathcal{O}(\epsilon t)$.
Main part of the complexity lies in finding an approximate LP solution.
Let $M(n)$ be the time to solve a system of $n$ linear equations. The running time of
max-min resource sharing is then in our case $\mathcal{O}(M(\frac{1}{\epsilon^2}) \frac{1}{\epsilon^4} + 
\frac{1}{\epsilon^7})$
(see \cite{jansen2004approximation}).
Therefore the running time of the Algorithm is $\mathcal{O}(M(\frac{1}{\epsilon^2}) \frac{1}{\epsilon^4} +
\epsilon t + \frac{1}{\epsilon^2}\log(\epsilon^2 t))$.
\section{Conclusion}
Based on approximate solutions, we developed an analogon to a theorem of Cook et al. \cite{cook1986sensitivity}. 
Our improvement
helps to develop online algorithms with a migration factor that is bounded by a polynomial
in $1/ \epsilon$, while algorithms based on Cook's theorem usually have exponential migration factors. 
We therefore applied our techniques to the famous online bin packing problem.
This led to the creation of the first fully robust AFPTAS for an NP-hard online optimization problem.
The migration factor of our algorithm is of size $\mathcal{O}(\frac{1}{\epsilon^4})$, which is a notable 
reduction compared to previous robust algorithms.
When a new item arrives at time $t$ the algorithm needs
running time of $\mathcal{O}(M(\frac{1}{\epsilon^2}) \frac{1}{\epsilon^4} +
\epsilon t + \frac{1}{\epsilon^2}\log(\epsilon^2 t))$, where $M(n)$ is the time to solve a system 
of $n$ linear equations.
Any improvement to the max-min resource sharing algorithm based on the special structure of bin packing would 
immediately speed up our online algorithm. We believe that there is room to reduce
the running time and the migration factor. Note for example that we give only a very rough bound for the
migration factor as the algorithm repacks $\mathcal{O}(\frac{1}{\epsilon^3})$ bins. Repacking these bins in a more
carefully way might lead to a smaller migration factor.
An open question is the existence of an 
AFPTAS with a constant migration factor that is 
independent of $\epsilon$.
We mention in closing that the LP/ILP-techniques presented are very general and hence can possibly be used
to obtain
fully robust algorithms for several other online optimization problems as well (i.e. multi-commodity flow, 
strip packing, scheduling with malleable/moldable task or scheduling with resource constraints).

\bibliographystyle{plain}
\bibliography{library}

\end{document}